\PassOptionsToPackage{dvipsnames,svgnames,x11names}{xcolor} 

\documentclass[acmsmall,screen, authorversion=true]{acmart}

\usepackage{amsmath}
\usepackage{amsthm}
\usepackage{subcaption}
\usepackage{supertabular}
\usepackage{stmaryrd}
\usepackage{xcolor}
\usepackage{multirow}
\usepackage{calc}
\usepackage[T1]{fontenc}
\usepackage{graphicx}
\usepackage{xspace}
\usepackage{dashrule}  
\usepackage{stackrel}
\usepackage{enumerate}
\usepackage{framed}
\usepackage{bussproofs}
\usepackage{mdframed}  
\usepackage{comment}
\usepackage{ifthen}
\usepackage{pdftexcmds}
\usepackage{mathpartir}
\usepackage{fancyvrb}
\usepackage{booktabs}
\usepackage{tikz}
\usepackage[zi4]{inconsolata-slanted} 
\usepackage{listings}
\usepackage{url}

\usepackage{lstparams}
\usepackage{lsthaskell,lstcoq,lstedit}

\newcommand{\codeplus}[3]{%
  \lstinputlisting[%
        #3,
        linerange={#2},
        rangebeginprefix=(*\ begin\ ,%
        rangebeginsuffix=\ *),%
        rangeendprefix=(*\ end\ ,%
        rangeendsuffix=\ *),
        includerangemarker=false]{#1}}

\newcommand\zcdc[1]{\let\par\endgraf\cdc{#1}}
\newcommand\eqRule[2]{\hbox{\zcdc{#1}}\ =\ \hbox{\zcdc{#2}}}
\newcommand\congRule[2]{\hbox{\zcdc{#1}}\ \cong\ \hbox{\zcdc{#2}}}

\newcommand\subseteqRule[2]{\hbox{\zcdc{#1}}\ \subseteq\ \hbox{\zcdc{#2}}}
\newcommand\den[1]{\llbracket#1\rrbracket}

\usepackage[capitalize, nameinlink]{cleveref}

\usepackage{siunitx}
\newcommand\cd[1]{\lstinline[breakatwhitespace]{#1}}
\newcommand\cdh[1]{\lstinline[language=Haskell,breakatwhitespace]{#1}}

\newcommand\cdc[1]{\lstinline[language=Coq,breakatwhitespace]{#1}}

\newcommand\hstocoq{\texttt{hs-to-coq}\xspace}

\newcommand\logic{Tl\"on embedding}
\newcommand\Logics{Tl\"on Embeddings}
\newcommand\logics{Tl\"on embeddings}
\newcommand\etc{\textit{etc.}}
\newcommand\ie{\textit{i.e.,\ }}
\newcommand\eg{\textit{e.g.,\ }}

\newcommand{\struct}{program adverb}
\newcommand{\structs}{program adverbs}

\newcommand{\structtype}{adverb data type}
\newcommand{\structtypes}{adverb data types}

\newcommand{\lan}{$\mathcal{B}$}

\newcommand{\socketnet}{\textsc{NetImp}}
\newcommand{\socketnetspec}{\textsc{NetSpec}}

\setcopyright{rightsretained}
\acmPrice{}
\acmDOI{10.1145/3547632}
\acmYear{2022}
\copyrightyear{2022}

\acmSubmissionID{icfp22main-p26-p}
\acmJournal{PACMPL}
\acmVolume{6}
\acmNumber{ICFP}
\acmArticle{101}
\acmMonth{8}

\begin{abstract}
  Free monads (and their variants) have become a popular general-purpose tool
  for representing the semantics of effectful programs in proof
  assistants. These data structures support the compositional definition of
  semantics parameterized by uninterpreted events, while admitting a rich
  equational theory of equivalence. But monads are not the only way to
  structure effectful computation, why should we limit ourselves?

  In this paper, inspired by applicative functors, selective functors, and other
  structures, we define a collection of data structures and theories, which we
  call \emph{program adverbs}, that capture a variety of computational
  patterns. Program adverbs are themselves composable, allowing them to be used
  to specify the semantics of languages with multiple computation patterns.  We
  use program adverbs as the basis for a new class of semantic embeddings
  called \emph{\logics}.  Compared with embeddings based on free
  monads, \logics{} allow more flexibility in computational modeling of effects,
  while retaining more information about the program's syntactic structure.
\end{abstract}



\begin{document}

\title{Program Adverbs and \Logics}

\newif\ifextended
\extendedtrue

\author{Yao Li}
\email{hnkfliyao@gmail.com}
\orcid{0000-0001-8720-883X}
\affiliation{%
  \department{Computer and Information Science}
  \institution{University of Pennsylvania}
  \streetaddress{3330 Walnut St}
  \city{Philadelphia}
  \state{PA}
  \postcode{19104}
  \country{USA}
}

\author{Stephanie Weirich}
\email{sweirich@cis.upenn.edu}
\orcid{0000-0002-6756-9168}
\affiliation{%
  \department{Computer and Information Science}
  \institution{University of Pennsylvania}
  \streetaddress{3330 Walnut St}
  \city{Philadelphia}
  \state{PA}
  \postcode{19104}
  \country{USA}
}

\maketitle

\bibliographystyle{ACM-Reference-Format}
\citestyle{acmauthoryear}   

\section{Introduction}\label{sec:intro}

Suppose that you want to formally verify a program written in your favorite
language---be it Verilog, Haskell, or C---your first step would be to translate
that program and a description of its semantics to a formal reasoning system,
such as Coq~\citep{coq}. This step is known as \emph{semantic
embedding}~\citep{embedding}.

There are multiple approaches to semantic embeddings. The two most well-known
were proposed by \citet{embedding}: \emph{shallow embeddings}, which represent
terms of the embedded language using equivalent terms of the embedding language,
and \emph{deep embeddings}, which represent terms using abstract syntax
trees~(ASTs) and represent their semantics via some interpretation function.

Shallow embeddings are convenient because they are simple, but they have their
limitations. It is impossible to use them to state and reason about properties
related to syntax, because they do not retain the syntactic structure of the
original program. Furthermore, shallow embeddings fix a single semantics, so
they are less robust to changes in program interpretations. Such edits require
changing the translation process, in addition to the semantic domain~(\ie the
type used for representing the semantics of the embedded language).

On the other hand, deep embeddings are more modular thanks to an extra layer
(\ie the AST) that defines the syntax of the embedded language. When we need to
change the semantics, we only need to change the \emph{interpretation} that maps
the AST into some semantic domain---the translation to the formal reasoning
system remains unchanged. Furthermore, the AST makes it possible to state and
prove properties related to the original program's syntactic structure. The
downside is that interpreting and reasoning about properties based on this AST
takes more effort than with shallow embeddings.

The pros and cons make it hard to choose between shallow and deep
embeddings. Fortunately, we don't need to commit to a single option. We can
use \emph{mixed embeddings}, a style of embedding that includes characteristics
of each. In this style, parts of a language are embedded ``shallowly'' while
other parts are embedded ``deeply''. However, in any mixed embedding, we must
ask: where should we draw the line to separate the shallowly embedded part from
the deeply embedded part?

Recent efforts have focused on mixed embeddings based on freer
monads~\citep{freer} or their variants~\citep{itree, mcbride-free, one-monad,
steelcore, resumption, delay}. The style has been shown useful for representing
and reasoning about effectful computation in various
applications~\citep{freespec2, reaching-star, vellvm-itree, itree-kv-server,
itree-hol, verify-effectful-haskell-in-coq, adam-binder}. Beyond these
applications, this style points out a useful guideline for answering the
question above.  That is: modeling the pure parts of the computation
``shallowly'' and the effectful parts ``deeply''.

Our work builds on this idea of separating pure and effectful parts in a mixed
embedding, but inspects the following question: Why freer monads? We find that
this is because freer monads model \emph{one} general computation pattern that
is common in many languages. However, the finding also implies that there are
other computation patterns not captured by freer monads.

Following this observation, we propose a new class of mixed embeddings
called \emph{\logics}.\footnote{The name \logic{} is a reference to the short
story \emph{Tl\"on, {U}qbar, {O}rbis {T}ertius} by Jorge Luis Borges. In the
short story, Tl\"on is an imaginary world, where its parent language does not
have any nouns, but only ``impersonal verbs, modified by monosyllabic suffixes
(or prefixes) with an adverbial value''~\citep{tlon}.}  \logics{} model programs
using structures called \emph{\structs}, which are reifications of familiar type
classes (\eg~\cdc{Applicative}, \cdc{Selective}, \cdc{Monad}) paired with
equational theories. Like freer monads, these free structures can be used to
combine shallowly embedded pure computation with deeply embedded computational
effects.  However, \structs{} provide choices in the semantics through the
selection of the structure and equational theory. For example, the
``statically'' adverb, based on applicative functors and their free theory,
models computation where control flow and data flow in the semantics are
fixed. Or, by modifying the equational theory of the free applicative structure
to include commutativity, we can describe computation that is ``statically and
in parallel''.

We make the following contributions:
\begin{itemize}
\item We compare the trade-offs of different styles of semantic
      embeddings in the context of formal reasoning and
      propose \logics~(\Cref{sec:motivation}).
\item We define \structs{} and show how to define their syntactic parts and
      their semantic parts~(\Cref{sec:statically}).
\item We refactor \structs{} to support composition and extension. We
      motivate why we want to compose \structs{} and define a composition
      algebra~(\Cref{sec:combine-structs}).
\item We implement composable \structs{} using the Coq proof
      assistant. A major challenge for implementing them in Coq is supporting
      extensible inductive data types~\citep{expression}. We show one way of
      addressing this challenge by adapting the \emph{Meta Theory \`{a} la
      Carte}~(MTC) approach~\citep{mtc}~(\Cref{sec:combine-structs}).
\item We identify five basic \structs{} from commonly used type classes in
      Haskell and we prove that these \structs{} are
      sound~(\Cref{sec:statically}). We also identify two add-on \structs{} that
      are used in combination with basic \structs~(\Cref{sec:combine-structs}).
\item We demonstrate the usefulness of \structs{} via three distinct language
      examples including a simple circuit language~(\Cref{sec:motivation}),
      Haxl~\citep{haxl}, and a networked server adapted
      from \citet{itree-server}~(\Cref{sec:case-study}).
\end{itemize}
Additionally, we discuss the choice of \structtypes{} we use and alternative
approaches to implement composable \structs{} in \cref{sec:discussion} and the
related work in \cref{sec:related-work}. We provide the Coq formalization of all
the key concepts, theorems, and examples shown in this paper in our
supplementary artifact~\citep{tlon-artifact}.

\section{Semantic Embeddings}\label{sec:motivation}\label{subsec:example}

In this section, we first demonstrate different forms of semantic embeddings
using a simple circuit language called \lan{} and compare how each form of
embedding can be used to reason about programs written in this language. To
distinguish the embedded language and the embedding language, we use
mathematical notation to describe \lan{} and use Coq code to describe its
embeddings.

The syntax of \lan{} appears in \cref{fig:lanb-syntax}. Semantically,
we want the Boolean operators to have their usual semantics. However,
\lan{} can read from the variables that represent references to
external devices and we don't want to fix those values in the semantics.
Furthermore, we don't know if the values are immutable: they might change over
time, or they might change after each read, \etc\@

\begin{figure}[t]
\begin{align*}
\emph{literals}\qquad b   &\ ::=\ \hbox{\zcdc{true}} \;|\; \hbox{\zcdc{false}} \\
\emph{terms}\qquad    t,u &\ ::=\ x \;|\; b \;|\; \neg t \;|\; t \wedge u \;|\; t \vee u
\end{align*}
\caption{The syntax of \lan.}\label{fig:lanb-syntax}
\end{figure}

The four embeddings that we consider in this section are defined in the right
column of \cref{fig:shallow-translation}. We use $\den{\cdot}_S$,
$\den{\cdot}_D$, $\den{\cdot}_M$, and $\den{\cdot}_A$ to represent the
translation from a term of \lan{} to shallow, deep, and two mixed embeddings,
respectively. These translations refer to the definitions in the left column as
well as to the standard classes and notations for functors, monads, \etc, shown
in \cref{fig:functor-and-ap}.

To compare embeddings, we will use each to consider the following questions
regarding the semantics of \lan{}:
\begin{enumerate}
\item Is $x$ equivalent to $x \wedge x$?
\item Is $x$ equivalent to $x \wedge \text{true}$?
\item Is $t \wedge u$ equivalent to $u \wedge t$?
\item Is the number of variable accesses always less than or equal to 2 to the
      power of the circuit's depth?
\end{enumerate}

Because we are modeling a circuit language that uses unknown external devices,
we don't want to be able to prove or disprove property (1). This property may
hold or not hold depending on the situation. If the external devices are
immutable, this property will be true. Otherwise, we may be able to falsify it.
In contrast, we would like our semantic embedding to give us tools to verify
properties (2) and (3) because these properties should hold regardless of the
properties of our external device. The former holds because on both sides of the
equivalence relation we have only accessed the variable $x$ once. The latter is
due to circuits run $\wedge$ in parallel---whatever result appears in $t \wedge
u$ can also appear in $u \wedge t$ and vice versa, regardless of what effects
could be involved in $t$ or $u$. The last property (4) relates a dynamic
property of the semantics (the number of variable accesses) to a syntactic
property of the circuit (the size of the circuit itself).

\begin{figure}
\begin{tabular}{ll}
\fbox{Shallow Embedding}\\
\begin{minipage}{0.5\linewidth}
\codeplus{ProgramAdverbs/Examples/Section2.v}{reader}{language=Coq,resetmargins=true}
\end{minipage}
&
\begin{minipage}{0.4\linewidth}
\[ \begin{array}{r@{\;}l}
\den{\hbox{\ensuremath{\cdot}}}_S &: \text{\cdc{Reader bool}}
\\ \den{\hbox{\cdc{true}}}_S &= \text{\cdc{ret true}}
\\ \den{\hbox{\cdc{false}}}_S &= \text{\cdc{ret false}}
\\ \den{x}_S &= \text{\cdc{ask x}}
\\ \den{\neg t}_S &= \text{\cdc{negb <\$>} } \den{t}_S
\\ \den{t \wedge u}_S &= \text{\cdc{t'<-} $\den{t}_S$\cdc{; u'<-} $\den{u}_S$\cdc{;}}
\\ &\quad \text{\cdc{ret (andb t' u')}}
\\ \den{t \vee u}_S &= \text{\cdc{t'<-} $\den{t}_S$\cdc{; u'<-} $\den{u}_S$\cdc{;}}
\\ &\quad \text{\cdc{ret (orb  t' u')}}
\end{array}\]
\end{minipage}
\\
\\
\begin{minipage}{0.45\linewidth}
\fbox{Deep Embedding}
\codeplus{ProgramAdverbs/Examples/Section2.v}{deep_embedding}{language=Coq,resetmargins=true}
\end{minipage}
&
\begin{minipage}{0.4\linewidth}
\[ \begin{array}{r@{\;}l}
\den{\hbox{\ensuremath{\cdot}}}_D &: \text{\cdc{term}}
\\  \den{\hbox{\cdc{true}}}_D &= \text{\cdc{Lit true}}
\\ \den{\hbox{\cdc{false}}}_D &= \text{\cdc{Lit false}}
\\ \den{x}_D &= \text{\cdc{Var x}}
\\ \den{\neg t}_D &= \text{\cdc{Neg} } \den{t}_D
\\ \den{t \wedge u}_D &= \text{\cdc{And} } \den{t}_D \text{ } \den{u}_D
\\ \den{t \vee u}_D &= \text{\cdc{Or} } \den{t}_D \text{ } \den{u}_D
\end{array}\]
\end{minipage}
\\
\\
\begin{minipage}{0.55\linewidth}
\fbox{Freer Monad Embedding}
\codeplus{ProgramAdverbs/Examples/Section2.v}{Freer}{language=Coq,resetmargins=true}
\end{minipage}
&
\begin{minipage}{0.4\linewidth}
\[ \begin{array}{r@{\;}l}
\den{\hbox{\ensuremath{\cdot}}}_M &: \text{\cdc{FreerMonad DataEff bool}}
\\ \den{\hbox{\cdc{true}}}_M &= \text{\cdc{Ret true}}
\\ \den{\hbox{\cdc{false}}}_M &= \text{\cdc{Ret false}}
\\ \den{x}_M &= \text{\cdc{Bind (GetData x) Ret}}
\\ \den{\neg t}_M &= \text{\cdc{negb <\$>} } \den{t}_M
\\ \den{t \wedge u}_M &= \text{\cdc{t'<-} $\den{t}_M$\cdc{; u'<-} $\den{u}_M$;}
\\ & \quad \text{\cdc{Ret (andb t' u')}}
\\ \den{t \vee u}_M &= \text{\cdc{t'<-} $\den{t}_M$\cdc{; u'<-} $\den{u}_M$\cdc{;}}
\\ & \quad \text{\cdc{Ret (orb  t' u')}}
\end{array}
\]
\end{minipage}
\\
\\
\begin{minipage}{0.55\linewidth}
\fbox{Reified Applicative Embedding}
\codeplus{ProgramAdverbs/Examples/Section2.v}{applicative}{language=Coq,resetmargins=true}
\end{minipage}
&
\begin{minipage}{0.4\linewidth}
\[ \begin{array}{r@{\;}l}
\den{\hbox{\ensuremath{\cdot}}}_A &: \text{\cdc{ReifiedApp DataEff bool}}
\\ \den{\hbox{\cdc{true}}}_A &= \text{\cdc{Pure true}}
\\ \den{\hbox{\cdc{false}}}_A &= \text{\cdc{Pure false}}
\\ \den{x}_A &= \text{\cdc{EmbedA (GetData x)}}
\\ \den{\neg t}_A &= \text{\cdc{negb <\$>} } \den{t}_A
\\ \den{t \wedge u}_A &= \text{\cdc{LiftA2} \cdc{andb} $\den{t}_A$  $\den{u}_A$}
\\ \den{t \vee u}_A &= \text{\cdc{LiftA2} \cdc{orb} $\den{t}_A$  $\den{u}_A$}
\end{array}\]
\end{minipage}
\end{tabular}
\caption{Semantic embeddings of \lan{} in Coq. We use the infix
operator \cdc{<\$>} to represent a functor's \cdc{fmap} method and a notation
similar to Haskell's \cdh{do} notation to represent monadic binds. The
functions \cdc{negb}, \cdc{andb}, and \cdc{orb} are Coq's functions defined on
the \cdc{bool} type.}\label{fig:reader-monad}\label{fig:shallow-translation}\label{fig:term}\label{fig:deep-translation}\label{fig:freer}\label{fig:mixed-translation}\label{fig:statically-adv}
\end{figure}

\begin{figure}[t]
\begin{tabular}{l}
\begin{lstlisting}[language=Coq]
Class Functor (F : Type -> Type) :=
  { fmap : forall {A B}, (A -> B) -> F A -> F B }.
Class Applicative (F : Type -> Type) `{Functor F} :=
  { pure   : forall {A}, A -> F A ;
    liftA2 : forall {A B C}, (A -> B -> C) -> F A -> F B -> F C }.
Class Selective (F : Type -> Type) `{Applicative F} :=
  { selectBy  : forall {A B C}, (A -> ((B -> C) + C)) -> F A -> F B -> F C }.
Class Monad (F : Type -> Type) `{Applicative F} :=
  { ret : forall {A}, A -> F A ;
    bind : forall {A B}, F A -> (A -> F B) -> F B }.
\end{lstlisting}
\\
\\
\emph{Default \cdc{fmap} definitions}\\
\codeplus{ProgramAdverbs/Examples/Section2.v}{fmap_monad}{language=Coq} \\
\codeplus{ProgramAdverbs/Examples/Section2.v}{fmap_ap}{language=Coq} \\
\\
\end{tabular}

\caption{Coq type classes for functors, applicative
  functors~\citep{applicative}, selective functors~\citep{selective}, and
  monads~\citep{moggi-monad, wadler-monad}, as well as default definitions
  of \cdc{fmap}.}\label{fig:functor-and-ap}\label{fig:selective-and-monad}\end{figure}

\subsection{A Shallow Embedding}

To use a shallow embedding to represent the semantics of \lan, we need a way to
represent the effects of reading from external devices---the most common way of
doing this is
using \emph{monads}~(\cref{fig:selective-and-monad}). But \emph{which} one? A
simple option is the reader monad~\citep{wadler-monad, mark95}. We show core
definitions of a specialized reader monad at the top left of
\cref{fig:reader-monad}.\footnote{For simplicity, we specialize the monad so
  that its environment has type \cdc{var -> bool}. The commonly used reader
  monad is more general that the type of its environment is parameterized.}  The
  translation from \lan{} to \cdc{Reader bool} is given in the same figure.
  Following the terminology used by \citet{deep-and-shallow}, we
  call \cdc{Reader bool} the
\emph{semantic domain} of our shallow embedding. Of course, the reader monad
is just one possible semantic domain, other candidates include Dijkstra
monads~\citep{dijkstra-monad}, predicate transformer
semantics~\citep{pt-semantics}, \etc\@

Using the reader monad, we can prove that property (1) is true, using
($\simeq_S$), the pointwise equality of functions. More specifically, we can
prove the following Coq theorem:
\begin{lstlisting}[language=Coq,mathescape=true]
forall x, ask x $\simeq_S$ x1 <- ask x; x2 <- ask x; ret (andb x1 x2)
\end{lstlisting}
We ``ask'' twice on the right hand side of the equivalence to model accessing
variable $x$ twice during program runtime. However, \cdc{x1} equals to \cdc{x2}
in our case since nothing has changed the global store. After proving that, the
theorem can be proved via a case analysis on \cdc{x1}.

However, note that our proof relies on ``nothing has changed the global store,''
but we don't know if this is true, as we don't know anything about the
characteristics of the external device. Indeed, property (1) should \emph{not}
be true if we have a device where its values change over time: the value of $x$
might change between two variable access. This is a problem with our choice of
semantic domain. By choosing the reader monad, we introduce more assumptions
over the semantics of \lan, which results in proving a property that is not
supposed to be true in the original language \lan.

Although this is not a problem with the approach of shallow embedding---we can
choose a different monad than the reader monad, the style does force us to
choose a concrete semantic domain early. In practice, we sometimes need to
change the semantic domain, either because we made a wrong assumption or
because the language evolves. With shallow embeddings, we would need to change
the entire translation process to change this domain.

Unlike property (1), property (2) is true even though we don't know anything
about the external device. This is because on both sides of the equivalence
relation we have only accessed the variable $x$ once.
\ifextended
Property (2) can be stated as follow with our shallow embedding:
\begin{lstlisting}[language=Coq,mathescape=true]
forall x, ask x $\simeq_S$ x1 <- ask x; ret (andb x1 true)
\end{lstlisting}
\fi
The proof follows from the theories of Coq's \cdc{bool} type and
the \cdc{Reader} monad. However, even though this property should be true
regardless of the external device, our mechanical proof still relies on the
assumption that the external device is immutable---this is again because the
property is stated in terms of the reader monad. If we change the shallow
embedding to use a different semantic domain, we would need to prove this
property again.

Property (3) is true and we can prove it to be true using our shallow embedding,
but that is just a lucky hit. Even though we know nothing about the external
device, there is an equivalence between $t \wedge u$ and $u \wedge t$ because
the two operands $t$ and $u$ run in parallel in a circuit. A proof based on our
shallow embedding would, on the other hand, be based on the wrong assumption
that the external device is immutable.

We cannot state property (4) with our shallow embedding. Our shallow embedding
does not retain the syntactic structure of the original program so we cannot
define a function that calculates the depth of the circuit.

\subsection{A Deep Embedding}

In a deep embedding, we first define an abstract syntax tree~(AST) for \lan. For
example, we can use the \cdc{term} data type shown in \cref{fig:term}. Our
translation from \lan{} to the \cdc{term} is shown in the same figure. Note that
the \cdc{term} data type does not encode \emph{any} semantic meaning.

Without an interpretation, we cannot prove any of the first three
properties. This is actually ideal for answering question (1) since we know
nothing about the external device so we should not be able to prove it (nor
should we be able to prove it wrong!). However, by leaving the entire syntax
tree uninterpreted we are now unable to prove property (2) or (3), either.

A way out of this quandary is to define a coarser \emph{equivalence relation}
for ASTs and use that relation in the statement of properties (2) and (3). For
example, we can interpret each \cdc{term} using the reader monad (as in the
shallow embedding) and use the pointwise equality for that type. The
proofs are essentially the same as the above.

One advantage of the deep embedding in this case is that, if we would like to
change our definition of equivalence, we can do so by choosing a different
\emph{interpretation} without changing the translation process. In other
words, deep embeddings achieve better modularity by introducing an intermediate
layer. The price, however, is that it takes effort to build that extra
intermediate layer.  This extra effort seems small here, but can become tedious
with some languages, \eg those with features like ``let'' that introduce
variable bindings~\citep{poplmark}.

However, we still face a similar problem with the shallow embedding: If we
would like to change the interpretation in our definition of equivalence, we
need to prove our properties again. This suggests that another intermediate
layer between deep and shallow embeddings might be helpful, as we will see
in the next subsection.

The primary benefit we have by using the deep embedding is that we can now state
and prove property (4). This is because the deep embedding gives us a
representation of the program's original syntactic structure. This allows us to
define the following function that counts the depth of a circuit:
\codeplus{ProgramAdverbs/Examples/Section2.v}{depth}{language=Coq}
Since we assume a straightforward semantics for \lan, the number of variable
access at runtime equals to the number of variables appeared in a \cdc{term}, so
we can directly prove property (4) by an induction over the \cdc{term} data
type.

\subsection{A Mixed Embedding Based on Freer Monads}\label{subsec:circuit-study}

A semantic embedding can be partially shallow and partially deep. We use the
term \emph{mixed embeddings} to describe embeddings with this property.  One
style of mixed embeddings that is popular today is based on \emph{freer
monads}~\citep{itree, mcbride-free, one-monad, steelcore, freespec2,
reaching-star, adam-binder}. In this type of mixed embeddings, the pure parts of
the program are embedded shallowly, while effects are embedded deeply (and
abstractly) using algebraic data types ``connected'' by freer monads.

The core definitions of freer monads are in the left column of
\cref{fig:freer}. The \cdc{FreerMonad} data type is parameterized by an
abstract effect \cdc{E} of \cdc{Type -> Type} and a return type \cdc{R}
of \cdc{Type}. Conceptually, it collects all the deeply embedded effects \cdc{E}
in a right-associative monadic structure.

For any effect type, \cdc{FreerMonad E} is a monad as demonstrated by the
\cdc{Ret} constructor and \cdc{bind} function.  The \cdc{bind} function
pattern matches its first argument \cdc{m} and, in the case of \cdc{Bind},
passes its second arguments \cdc{k} to the continuation of \cdc{m}. This
``smart constructor'' ensures that binds always associate to the right.

To embed \lan{}, we model reading data from external devices using the effect
type \cdc{DataEff}. This datatype includes only one (abstract) effect,
called \cdc{GetData}. This constructor represents a data retrieval with the
variable \cdc{v : var} that returns an unknown \cdc{bool}. Similar to how
the \cdc{term} data type says nothing about the semantics of \lan{}, the effect
data type \cdc{DataEff} says nothing about the semantics of a data read. As a
result, we say that the effects are embedded deeply in this style.

The embedding function appears on the right side of
\cref{fig:mixed-translation}. The translation strategy is almost the same as
embedding \lan{} using the reader monad. The only exception is in the variable
case~(the effectful part): here the \cdc{Bind} constructor marks the
occurrence of the \cdc{GetData} effect.

In this mixed embedding, the pure parts of a \lan{} program have been translated
to a shallow semantic domain, but the effectful parts remain abstract. It turns
out that this separation is useful for both questions (1) and (2).

For question (1), we cannot answer it. This is desirable since we don't know if
it's true without knowing more about the external device.

\begin{figure}[t]
\begin{align*}
\textsc{Left identity} \quad  : \quad & \eqRule{ret a >>= h}{h a} \\
\textsc{Right identity} \quad  : \quad & \eqRule{m >>= ret}{m} \\
\textsc{Associativity} \quad  : \quad & \eqRule{(m >>= g) >>= h}{m >>= (fun x => g x >>= h)}
\end{align*}
\caption{The monad laws. The \cdc{>>=} symbol is the infix operator
for \cdc{bind}.}\label{fig:monad-laws}
\end{figure}

We can prove that property (2) is true even though the read effect is not
interpreted---this is because the property follows from the monad
laws~(\cref{fig:monad-laws}). However, we cannot prove property (3) because the
commutativity law is not one of the monad laws.

Ideally, we would also like to state and prove property (4). However, the
dynamic nature of freer monads forbids us from statically inspecting the
syntactic structure of the program. Interpreting the embedding does not help us,
either, since that would not preserve the original syntactic structure.

Our success with questions (1) and (2) suggests that we have found an useful
intermediate layer between shallow and deep embeddings, but our failure in
stating or proving properties (3) and (4) indicates that we haven't yet found
the most suitable representation for this circuit language.

\subsection{Another Mixed Embedding Based on Reified Applicative Functors}\label{sec:applicative-embedding}

The last embedding shown in the figure uses a type that reifies the interface
of \emph{applicative functors}~(\cref{fig:functor-and-ap}). As in freer monads,
this datatype is parameterized by deeply embedded abstract effects. These
effects, of type \cdc{E R}, are recorded by the \cdc{EmbedA} data constructor.

However, instead of constructors for \cdc{ret} and \cdc{bind}, this datatype
includes constructors for \cdc{pure} and \cdc{liftA2}, the two operations that
define applicative functors.\footnote{Alternatively, \cdc{Applicative} can also
be defined by \cdh{pure} and another operation \cdh{<*>} of type \cdc{F (A -> B)
-> F A -> F B}, where \cdc{F} is an \cdc{Applicative} instance. These two
definitions are equivalent, as we can derive the definition of \cdc{<*>}
from \cdc{liftA2} and vice versa.} The \cdc{Pure} constructor shallowly
``embeds'' a pure computation into the domain, and
\cdc{LiftA2} ``connects'' two computations that potentially contain effect
invocations. These constructors provide a trivial implementation of the
\cd{Applicative} type class for this datatype.

The translation of \lan{} to this datatype uses a deep embedding of variable
reads, using the \cdc{EmbedA} data constructor with the \cd{DataEff} type from
the previous embedding. Because, as in freer monads, this effect is modeled
abstractly, we cannot prove or disprove (1).

The translation function uses the applicative interface in the datatype to
translate the constants, unary and binary operators. These components are
modeled shallowly (\ie as Boolean constants and operators), but the program's
syntactic structure is retained by the translation. However, because of the
retainment, we need an additional equivalence relation to equate semantically
equivalent terms that are not syntactically equal. To prove (2), we include the
right identity law of applicative functors in the equivalence~(denoted by
$\cong$):
\begin{mathpar}
\inferrule{\forall \hbox{\zcdc{y}},\ \hbox{\zcdc{(fun _ x => x) a y = f a y}}}{\congRule{liftA2 f (pure a) b}{b}}{}
\end{mathpar}
This law is sufficient to show that (2) holds.

To model the parallelism of circuits, we could include the commutativity law in
the equivalence:
$$\congRule{liftA2 f a b}{liftA2 (flip f) b a}$$ This is sufficient to show
(3). Note that this is not one of the applicative functor laws. We defer showing
the soundness of including this rule in the equivalence
to \cref{subsec:in-parallel}.

This embedding also preserves enough of the syntax of the original program to
prove (4). To do so, we must first calculate the depth of circuits and the
number of variables under this encoding.
\codeplus{ProgramAdverbs/Examples/Section2.v}{app_depth}{language=Coq}
\ifextended
\codeplus{ProgramAdverbs/Examples/Section2.v}{app_numVar}{language=Coq}
\else
We omit the function that counts the number of variables as it is similar
to \cdc{app_depth}.
\fi
Then we can formalize (4) in Coq as follows:
\begin{lstlisting}[language=Coq]
Theorem heightAndVar : forall (c : ReifiedApp DataEff bool),
    app_numVar c <= Nat.pow 2 (app_depth c).
\end{lstlisting}
The theorem is provable by an induction over \cdc{c}.

\ifextended
This theorem is not coincidentally true. We can also prove that for any circuit
in the image of the encoding it has the same depth and number of variables as in
the deep embedding.
\fi

Furthermore, this embedding also allows us to reason about semantic properties
that depend on syntactic structures of circuits. One example is a semantics with
some cost model. In the semantics, we may not want our equivalence to equate,
for example, $x \wedge y \wedge z \wedge w$ and $(x \wedge y) \wedge (z \wedge
w)$ because they are not equivalent in their costs when parallelization is
present. Indeed, we cannot show that they are equivalent with our embedding due
to the absence of associativity in our equivalence.

\subsection{\logics{}}

Just as the reader monad models \emph{one} particular effect, freer monads
model \emph{one} particular computation pattern. Unfortunately, that
particular computation pattern is not suitable for our \lan{} example, because
it does not model parallel computation~(\ie property (3)), nor does it capture
the static data and control flows~(\ie property (4)). Instead we saw that the
mixed embedding in the previous subsection, based on reified applicative
functors, is a better approach.

Can we generalize the key idea even further? If we go beyond \lan, we might need
to model other computation patterns. Are there other mixed embeddings that would
be suitable for these tasks? How might we derive them?

To that end, we identify a novel set of mixed embeddings that we
call \emph{\logics{}}. The goal of these embeddings is to provide flexibility in
our models of effectful computation.\footnote{Here, we define effects as
communications with external environment that are performed by some explicit
operations. For example, \emph{mutable states} are effects which can be
explicitly incurred by operations such as \cdc{get} and \cdc{set}. For the same
reason, we also consider I/O~(with operations
like \cdc{read}, \cdc{print},~\etc) and exceptions~(with operations
like \cdc{throw},~\etc) as effects.} We define \logics{} by identifying a set
of \emph{\structs{}} that specify the embedding type and equational theory used
in the embedding. For example, the embedding in \cref{sec:applicative-embedding}
is based on an adverb composed of the \cd{ReifiedApp} type and \emph{some} rules
of commutative applicative functors.

The flexibility that \structs{} provide can perhaps be understood by comparing
them with effects: effects \emph{do} certain actions, and \structs{} model
\emph{how} these actions are done---similar to the difference between verbs
and adverbs. For example, the adverb we used in \cref{sec:applicative-embedding}
is called ``statically and in parallel'', which states that there is a static
dependency between different effect invocations and some of these effect
invocations are executed in parallel.

In the next section, we define our set of program adverbs more precisely and
discuss the reasoning principles that they provide for effectful computation.

\section{Program Adverbs}\label{sec:statically}

Program adverbs are the building blocks of \logics{}. Mathematically, they are
composed of two parts: a syntactic part, called the \structtype, and a semantic
part, called the adverb theory. More formally, we define \structs{} as
follows:\footnote{The Coq code of all definitions and theorems shown in this
section can also be found in our supplementary artifact~\citep{tlon-artifact}.}

\begin{definition}[Program Adverb]
A program adverb is a pair $(D, \cong_D)$. $D$ is called the adverb data type
and is parameterized by an effect $E$ and a return type $R$. The $\cong_D$
operation is called the adverb theory of $D$. It is a binary operation that
defines an equivalence relation on $D(E, R)$ for any $E$ and $R$.
\end{definition}
In the rest of the paper, we abbreviate $\cong_D$ as $\cong$ when $D$ is clear
from the context.

In Coq terms, an \structtype{} \cdc{D} has the type \cdc{(Type -> Type) -> Type
-> Type}. The first parameter of \cdc{Type -> Type} is the effect $E$ and it's
parameterized by its own return type; the second parameter is the return type
$R$. The adverb theory $\cong$ is a typed binary relation.\footnote{In addition
to equivalence relations, we can also define refinement relations
on \structs. We will show in \cref{subsec:addons} some adverbs with refinement
relations, but equivalence relations would suffice for most adverbs, so we only
include them in the core definitions of adverb theories. Refinement relations
can be added on demand.} More concretely:
\begin{lstlisting}[language=Coq]
Class Adverb (D : (Type -> Type) -> Type -> Type) :=
  { Equiv {E R} : relation (D E R) ;
    equiv {E R} : Equivalence (@Equiv E R) }.
Notation "a ≅ b" := (Equiv a b).
\end{lstlisting}
where \cdc{D} is the \structtype, \cdc{Equiv} is the adverb theory $\cong$,
and \cdc{equiv} is a proof showing that \cdc{Equiv} is an equivalence
relation. The datatype \cdc{relation} is defined as:
\begin{lstlisting}[language=Coq]
Definition relation (A : Type) := A -> A -> Prop.
\end{lstlisting}

This definition is overly general, so we focus our attention only on program
adverbs that are \emph{sound} according to the definition that we will develop
below. Furthermore, in this paper we will only consider adverbs defined by
reifying classes of functors.

\subsection{Adverb Data Types and Theories}

\begin{figure}[t]
\codeplus{ProgramAdverbs/Adverb/Streamingly.v}{streamingly_adv}{language=Coq}
\codeplus{ProgramAdverbs/Adverb/Statically.v}{statically_adv}{language=Coq}
\codeplus{ProgramAdverbs/Adverb/Conditionally.v}{conditionally_adv}{language=Coq}
\codeplus{ProgramAdverbs/Adverb/Dynamically.v}{dynamically_adv}{language=Coq}
\caption{The \structtypes{}}\label{fig:advs}
\end{figure}

The four key adverb data types, shown in \cref{fig:advs}, are derived from the
four type classes shown in \cref{fig:functor-and-ap}. We have already seen one
before in the applicative embedding in \cref{fig:statically-adv}. Other
definitions follow a similar pattern: the constructors of each data type include
one for embedding effects (of type \cd{E R}) and a constructor that reifies the
interface of each method of the type class.

In addition to an \structtype, every \struct{} also comes with some theories,
defined by an equivalence relation $\cong$. The purpose of the $\cong$ relation
is to equate all computations that are semantically equivalent regardless of
what effects are present.

For example, an adverb called \cdc{Statically} is composed of the
\cd{ReifiedApp} datatype with an equational theory based on three sorts of
rules: (1)~a congruence rule with respect to \cdc{LiftA2}, (2)~the laws of
applicative functors~\citep{applicative}, and (3)~the equivalence
properties~(\ie~reflexivity, symmetry, transitivity). We show the concrete rules
in \cref{fig:ap-equiv-rules}.

\begin{figure}[t]
\textbf{Congruence Rule}
\begin{mathpar}
\textsc{Congruence} \quad : \quad \inferrule{\congRule{a1}{a2} \\ \congRule{b1}{b2}}{\congRule{liftA2 f a1 b1}{liftA2 f a2 b2}}{}
\end{mathpar}
\\
\textbf{Applicative Functor Laws}
\begin{align*}
\textsc{Left Identity} \quad  : \quad & \inferrule{\forall \hbox{\zcdc{y}},\ \hbox{\zcdc{(fun _ x => x) a y = f a y}}}{\congRule{liftA2 f (pure a) b}{b}}{} \\
\textsc{Right Identity} \quad  : \quad & \inferrule{\forall \hbox{\zcdc{x}},\ \hbox{\zcdc{(fun x _ => x) x b = f x b}}}{\congRule{liftA2 f a (pure b)}{a}}{} \\
\textsc{Associativity} \quad  : \quad & \inferrule{\forall \hbox{\zcdc{x y z}},\ \hbox{\zcdc{f x y z = g y z x}}}{\congRule{liftA2 id (liftA2 f a b) c}{liftA2 (flip id) a (liftA2 g b c)}}{} \\
\textsc{Naturality} \quad  : \quad & \inferrule{\forall \hbox{\zcdc{x y z}},\ \hbox{\zcdc{p (q x y) z = f x (g y z)}}}{\congRule{liftA2 p (liftA2 q a b)}{liftA2 f a . liftA2 g b}}{}
\end{align*}
\\
\textbf{Equivalence Properties}
\begin{mathpar}
\textsc{Reflexivity} \quad : \quad \inferrule{\ }{\congRule{a}{a}}{} \qquad
\textsc{Symmetry} \quad : \quad \inferrule{\congRule{a}{b}}{\congRule{b}{a}}{} \\
\textsc{Transitivity} \quad : \quad \inferrule{\congRule{a}{b} \\ \congRule{b}{c}}{\congRule{a}{c}}{}
\end{mathpar}
\caption{The equivalence relation for \cdc{ReifiedApp}. The infix operator \cdc{.} denotes function compositions.}\label{fig:ap-equiv-rules}
\end{figure}

Why do we call this adverb \cdc{Statically}? The data dependency in
the \cdc{LiftA2} constructor of \cdc{ReifiedApp} shows that the data type
imposes a ``static'' data flow and control flow on the computation: we will
always need to run both parameters of type \cdc{ReifiedApp E A}
and \cdc{ReifiedApp E B} to get the result of type \cdc{ReifiedApp E C}, \ie we
cannot skip either computation. In addition, neither of the two parameters
depends on the result of the other, which allows us to statically inspect either
of them without running the other.

\paragraph{Remark}

The \structtypes{} and their associated theories form free structures similar to
those in \citet{free-ap, freer, selective, selective-impl}. However, one
distinction is that we intentionally do not normalize the \structtypes{} to
preserve syntactic structures. To distinguish un-normalized free structures and
normalized free structures, we use the term \emph{reified} structures to
describe the former and the term free structures to exclusively describe the
latter.  We defer the detailed comparison and trade-offs between reified
structures and free structures to \cref{sec:compare-freer}.

\subsection{Adverb Simulation}

One important property of \cdc{ReifiedApp} is that it can be interpreted to any
other instance of the \cdc{Applicative} class, as long as its embedded effects
can be interpreted to that instance. We can show this via the abstract
interpreter \cdc{interpA} shown in \cref{fig:interpA}. The interpreter shows
that given \emph{any} effect \cdc{E} and \emph{any} instance \cdc{I}
of \cdc{Applicative}, as long as we can find an effect interpretation
from \cdc{E A} to \cdc{I A} for any type \cdc{A}, we can interpret
a \cdc{ReifiedApp E A} to an \cdc{I A} for any type \cdc{A}.

\begin{figure}[t]
\codeplus{ProgramAdverbs/Adverb/Statically.v}{interpA}{language=Coq}
\caption{The interpretation from \cdc{ReifiedApp} to any instance of
the \cdc{Applicative} type class.}\label{fig:interpA}
\end{figure}

For example, we can interpret a \cdc{ReifiedApp DataEff} to the reader
applicative functor~(\cref{fig:reader-monad})\footnote{Every monad is also an
applicative functor, so the reader monad is also a reader applicative functor.}
by supplying the following function to the parameter \cdc{interpE}
of \cdc{interpA}:
\begin{lstlisting}[language=Coq]
Definition interpDataEff {A : Type} (e : DataEff A) : Reader A :=
  match e with GetData v => ask v end.
\end{lstlisting}
Similarly, we can interpret \cdc{ReifiedApp DataEff} to other semantic
domains that are applicative functors.

Why do we care if \cdc{ReifiedApp} can be interpreted into any instance
of \cdc{Applicative}? This is because different instances of \cdc{Applicative}
model different effects---if we have a data structure that can be interpreted to
all instances, we can develop a theory of it that can be used for reasoning
about properties that are true regardless of what effects are present.

To make the relation between an \structtype{} like \cdc{ReifiedApp} and
a class of functors like \cdc{Applicative} more precise, we define the
following \emph{adverb simulation} relation:
\begin{definition}[Adverb Simulation]
Given an \structtype{} $D$, a class of functors $C$, and a transformer $T$ on all
instances of $C$, we say that there is an adverb simulation from $D$ to $C$
under $T$, written $D \models_T C$, if we can define a function that, for any
effect type $E$, instance $F$ of type class $C$, and interpreter $f$ from $E(A)$
to $F(A)$ for any type $A$, interprets a value of $D(E, A)$ to $T(F)(A)$ for any
type $A$.
\end{definition}
We add some flexibility to this definition by making it parameterize over a
transformer $T$---we do not need this extra flexibility for now, but we will see
why it is useful in \cref{subsec:in-parallel}.

We also define an \emph{adverb interpretation} as follows:
\begin{definition}[Adverb Interpretation]
Given an \structtype{} $D$, a class of functors $C$, and a transformer $T$ on all
instances of $C$, an interpreter $I$ that shows $D \models_T C$ is called an
adverb interpretation, and we write $I \in D \models_T C$.
\end{definition}

Our \cdc{interpA} in \cref{fig:interpA} is an adverb interpretation. More
specifically, we say that
$$\hbox{\cdc{interpA}} \in \hbox{\cdc{ReifiedApp}} \models_{\hbox{\cdc{IdT}}} \hbox{\cdc{Applicative}}$$
where the \cdc{IdT} transformer is an identity \cdc{Applicative} transformer
that ``does nothing''. In the rest of the paper, when we have
$D \models_{\hbox{\cdc{IdT}}} C$ for any $D$ and $C$, we abbreviate it as
$D \models C$.

\subsection{Sound Adverb Theories}

To know that our adverb theory is \emph{sound}, \ie it doesn't equate
computations that are not semantically equivalent, we define the following
soundness property of adverb theories:
\begin{definition}[Soundness of Adverb Theories]
Given a \struct{} $(D, \cong)$ and an adverb interpretation $I \in D \models_T
C$, we say that the adverb theory $\cong$ is sound with respect to $I$ if there
exists a lawful equivalence relation $\equiv$ such that for all $d_1, d_2 \in
D$,
$$d_1 \cong d_2 \implies I(d_1) \equiv I(d_2).$$
\end{definition}
Let us use \cdc{idT} for the transformer $T$ for the moment. The equivalence
relation $\equiv$ on $C$ is lawful if they respect the congruence laws and the
class laws of $C$. For \cdc{Applicative}, we use the common applicative functor
laws regarding $\equiv$. Based on the soundness of adverb theories, we can
define the following soundness property of program adverbs with respect to their
adverb interpretations:
\begin{definition}[Soundness of Program Adverbs]
Given a \struct{} $(D, \cong)$ and an adverb interpretation $I \in D \models_T
C$, we say that the adverb is sound if the $\cong$ relation is sound with
respect to $I$.
\end{definition}

We can now prove that the \cdc{Statically} adverb is sound:
\begin{theorem}\label{thm:sound-statically}
The \cdc{Statically} adverb $(\hbox{\zcdc{RefieidApp}}, \cong)$ is sound with
respect to the adverb interpretation
$\hbox{\cdc{interpA}} \in \hbox{\cdc{ReifiedApp}} \models \hbox{\cdc{Applicative}}$.
\end{theorem}
\begin{proof}
By induction over the $\cong$ relation.
\end{proof}

\subsection{``Statically and in Parallel''}\label{subsec:in-parallel}

Two adverbs can use the same data type yet differ in their theories. Let's look
at a variant of the \cdc{Statically} adverb called
\cdc{StaticallyInParallel}. As its name suggests, it adds parallelization to a
static computation pattern.

Recall that the two computations connected by \cdc{liftA2} do not depend on each
other. This suggests that an implementation of \cdc{liftA2} can choose to run
them in parallel. Indeed, that observation is one of the key ideas behind
Haxl~\citep{haxl}.

Based on this idea, we also define the \cdc{StaticallyInParallel}
adverb. The \structtype{} of this adverb is the same as that
of \cdc{Statically}. However, its theory differs from \cdc{Statically} in the
following ways: (1)~it adds the commutativity rule: $$\congRule{liftA2 f a
b}{liftA2 (flip f) b a}$$ and (2)~it does not include the associativity and
naturality rules~(\cref{fig:ap-equiv-rules}).

The addition of commutativity rule states that the order that effects are
invoked does not matter. Note that compared with other rules, the commutativity
rule is not satisfied by every applicative functor. This might suggest that we
should not add it to the theory, as it might be a theory that only holds for
certain effects. Nevertheless, we can prove the soundness of the adverb theory
with respect to the following adverb simulation:
$$\hbox{\cdc{ReifiedApp}} \models_{\hbox{\cdc{PowerSet}}} \hbox{\cdc{Applicative}}$$
The \cdc{PowerSet} transformer is \emph{a transformer on applicative functors}
and its core definitions are shown in \cref{fig:nd_trans}. The key
of \cdc{PowerSet} is the \cdc{liftA2PowerSet} operation. When executed, it
creates two nondeterministic branches~(indicated by the disjunction \cdc{\\/}):
on one branch, it computes \cdc{a' : I A} before \cdc{b' : I B}, and vice versa
on the other branch. Intuitively, this is to model the nondeterministic
execution order in a parallel evaluation. Many of these operations depend on
$\equiv$, which is the lawful $\equiv$ relation on \cdc{I}.

\begin{figure}[t]
\begin{lstlisting}[language=Coq]
Definition PowerSet (I : Type -> Type) (A : Type) := I A -> Prop.

Definition embedPowerSet {A : Type} (a : I A) : PowerSet I A := fun r => r ≡ a.

Definition purePowerSet {A : Type} (a : A) : PowerSet I A := fun r => r ≡ pure a.

Definition liftA2PowerSet {A B C} (f : A -> B -> C)
                     (a : PowerSet I A) (b : PowerSet I B) : PowerSet I C :=
  fun r => exists a', a a' /\ exists b', b b' /\
           (liftA2 f a' b' ≡ r \/ liftA2 (flip f) b' a' ≡ r).

Definition EqPowerSet {A} : relation (PowerSet I A) :=
  fun p q => forall a, p a <-> q a.
\end{lstlisting}
\caption{The core definitions of a powerset applicative functor
transformer.}\label{fig:nd_trans}
\end{figure}

\begin{lemma}\label{lem:nd}
If $\equiv$ is a lawful equivalence relation
on \cdc{Applicative}, \cdc{EqPowerSet} is an equivalence relation
on \cdc{PowerSet I} that satisfies congruence, left identity, right identity,
and commutativity laws.
\end{lemma}
\begin{proof}
By definition.
\end{proof}
Note that \cdc{EqPowerSet I} does not satisfy the associativity or naturality
laws. Consider that we have \cdc{liftA2PowerSet id (liftA2PowerSet f a b) c},
for some \cdc{f}, \cdc{a}, \cdc{b}, and \cdc{c}: one of the possible evaluations
in this powerset is \cdc{liftA2 id (liftA2 (flip f) b a) c}, which does not
belong to the powerset of \cdc{liftA2PowerSet (flip id) a (liftA2PowerSet g b
c)}, for some \cdc{g} that is equivalent to \cdc{flip f}. The case for
naturality is similar. For this reason, we do not include these two rules in
$\cong$. We do not know if there exists an alternative nontrivial transformer
with an equivalence relation that satisfies \emph{all} the applicative laws in
addition to commutativity.

Nevertheless, we can show the following theorem with the help of \cref{lem:nd}:
\begin{theorem}
The adverb is sound:
$\hbox{\cdc{ReifiedApp}} \models_{\hbox{\cdc{PowerSet}}} \hbox{\cdc{Applicative}}$.
\end{theorem}
\begin{proof}
We can construct an
$\hbox{\cdc{interpPowerSet}} \in \hbox{\cdc{ReifiedApp}} \models_{\hbox{\cdc{PowerSet}}} \hbox{\cdc{Applicative}}$
by modifying \cdc{interpA}~(\cref{fig:interpA}) so that it
uses \cdc{embedPowerSet} on the \cdc{EmbedA} case, \cdc{purePowerSet} on
the \cdc{Pure} case, and \cdc{liftA2PowerSet} on the \cdc{LiftA2} case. With the
help of \cref{lem:nd}, we can show that for all $d_1, d_2 \in$ \cdc{ReifiedApp},
$$d_1 \cong
d_2 \implies \hbox{\cdc{interpPowerSet}}(d_1) \equiv \hbox{\cdc{interpPowerSet}}(d_2)$$
where $\equiv$ is \cdc{EqPowerSet}.
\end{proof}

Intuitively, we can define \cdc{StaticallyInParallel} as an adverb because, even
though with an effect running computations in different order might return
different results, a language can be implemented in a parallel way such that the
difference in evaluation orders is no longer observable.

The lack of associativity and naturality rules in the theory
of \cdc{StaticallyInParallel} might initially sound limiting, but, as we have
shown in the end of \cref{sec:applicative-embedding}, it turns out to be
desirable for applications like circuits.

\subsection{Other Basic Adverbs}

Besides \cdc{Statically} and \cdc{StaticallyInParallel}, we also identify three
other basic adverbs, namely \cdc{Streamingly}, \cdc{Conditionally},
and \cdc{Dynamically}, defined using the \structtypes{} in \cref{fig:advs}.

\paragraph{\cdc{Streamingly}.}
This \struct{} simulates \cdc{Functor} under \cdc{IdT}. The most simple form of
stream processing computes the data directly as it is received. This is captured
by the \cdc{fmap} interface~(\cref{fig:functor-and-ap}).

\paragraph{\cdc{Dynamically}.}
This adverb simulates \cdc{Monad}~(\cref{fig:selective-and-monad}). A monad is
the most expressive and dynamic among all four classes of functors thanks to its
core operation \cdc{bind}. Any kind of computation can happen in the second
operand and we can't know it without knowing a value of type \cdc{A}, which we
can only get by running the first operand.
This \struct{} is commonly used in representing many programming language for
its expressiveness, but it also allows for the least amount of static
reasoning.

Unlike \cdc{Statically}, this variant does not have an \cdc{InParallel}
variant. This might be surprising because there are many commutative
monads. However, those monads are commutative because their specific effects are
commutative. We cannot define a general powerset \emph{monad transformer} that
can make any monad satisfy the commutativity law.

\paragraph{\cdc{Conditionally}.}
We use this adverb to model conditional execution.
The definition of its \structtype{} is shown in \cref{fig:advs}. It reifies
the \cdc{Selective} type class~(\cref{fig:selective-and-monad}). The signature
operation of \cdc{Selective} is the \cdc{selectBy} operation. Loosely,
``applying'' a function of type \cdc{A -> ((B -> C) + C)} to a computation of
type \cdc{F A} gets you either \cdc{F (B -> C)} or \cdc{F C}. In the first case,
you will need to run the computation of type \cdc{F B}. You don't \emph{need} to
run the computation of type \cdc{F B} in the second case, but you can still
choose to run it.

Because we can encode conditional execution with this adverb, it is more
expressive than \cdc{Statically}. However, the extra expressiveness also makes
static analysis less accurate. Since we cannot know statically if the
computation \cdc{F B} in \cdc{selectBy} is executed, we can only get an
under-approximation~(assuming that \cdc{F B} is not executed) and an
over-approximation~(assuming that \cdc{F B} is executed) of the effects that
would happen, but not an exact set.

Even though we derive this adverb by reifying \cdc{Selective}, we do not wish to
model the adverb's theory using the laws of selective functors. This is
because the laws of selective functors do not distinguish them from applicative
functors. Indeed, every applicative functor is also a selective functor~(by
running the second argument even when not required) and vice versa, so
adhering to the ``default'' laws do not allow us to prove more
properties. Therefore, we add one simple rule to the selective functor laws:
$$\congRule{select (inr <$> a) b}{a}$$ The function \cdc{select} has the
type \cdc{F (A + B) -> F (A -> B) -> F B}, where \cdc{F} is an instance
of \cdc{Selective}. It is equivalent to
\begin{lstlisting}[language=Coq]
selectBy (fun x => match x with
                   | inl x => inl (fun y => y x)
                   | inr x => inr x
                   end).
\end{lstlisting}

This rule forces \cdc{select} to ignore the second argument when it does not
need run. However, we can no longer show that \cdc{Conditionally} adverb
simulates \cdc{Selective} by adding this laws, because $\cong$ is no longer an
under-approximation of $\equiv$. Instead, we show the following adverb
simulation:
$$\hbox{\cdc{ReifiedSelective}} \models \hbox{\cdc{Monad}}$$

In this way, \cdc{Conditionally} serves as a compromise between \cdc{Statically}
and \cdc{Dynamically}. Its \structtype{} is more similar to \cdc{Statically} and
allows for some static analysis, while its theories are more similar
to \cdc{Dynamically}.

\section{Composable Program Adverbs}\label{sec:combine-structs}

From a monad instance, we can derive an applicative functor instance. From an
applicative functor instance, we can derive a functor instance. We can derive a
selective instance from an applicative functor and vice versa.\footnote{This is
one special thing about selective functors: every selective functor is an
applicative functor and the reverse is also true. However, separating these two
classes is still useful because the automatically derived instances might not be
what we want, as discussed in \citet{selective}.} This subsumption hierarchy
among classes of functors means that we can choose the most expressive abstract
interface of a data type, and that choice automatically includes the less
expressive interfaces.

However, although we can derive a ``default'' applicative functor from a monad,
we don't always want to do that---\eg we may want to define a different behavior
for \cdc{liftA2} than the one derived from \cdc{bind}. Indeed, Haxl is one such
example, where \cdc{bind} is defined as a sequential operation and \cdc{liftA2}
is parallel so that certain tasks with no data dependencies can be automatically
parallelized~\citep{haxl}. In the \structs{} terminology, the semantics of their
language is composed of a ``statically and in parallel'' adverb and a
``dynamically'' adverb.

In addition, some languages may have a part that corresponds to the
``statically'' adverb and some extensions that correspond to ``dynamically''. If
we only use the ``dynamically'' adverb to reason about programs written in this
language, we lose the ability to state properties for the ``statically'' subset.

We need a way to compose multiple \structs. Therefore, in this section, we
refactor \structs{} to \emph{composable \structs}.

\subsection{Uniform Treatment of Effects and Program Adverbs}\label{subsec:uniform}

Effects are commonly considered secondary to monads. This treatment of effects
carries over to the approaches based on freer monads and our previous
implementation of \structs, where the effects are a parameter of \structtypes.

This approach works well when we use one fixed \struct, but needs an update when
multiple adverbs are involved. This is because, in both scenarios we mentioned
earlier, our intention is not to combine \structs{} that each contain their own
set of effects---we would like the composed \structs{} to share the same set of
effects. One solution is requiring that we can only join \structs{} when they
share the same set of effects, but that would require extra machinery.

In our work, we choose to give a uniform treatment to effects
and \structs. \Cref{fig:structs-algebra} shows our algebra for effects
and \structs. The algebra includes an $\oplus$ operator which is
a \emph{disjoint union} of effects \emph{and} \structtypes. We define an
equivalence relation $\approx$ on effects and \structtypes{} as follows: for all
$A, B$ that are effects and \structtypes, $A \approx B$ if there exists
a \emph{bijection} between $A$ and $B$. Similarly, we define an $\uplus$
operator for the disjoint union of adverb theories. We define an equivalence
relation $\Leftrightarrow$ on adverb theories as follows: for any \structtype{}
$D$ and adverb theories $P, Q$, which are adverb theories of $D$,
$P \Leftrightarrow Q$ if $a\ P\ b \Longleftrightarrow a\ Q\ b$ for all $a, b \in
D$, where $\Longleftrightarrow$ is the logical symbol for ``if and only
if''. Properties of this algebra are also shown in \cref{fig:structs-algebra}.

\begin{figure}[t]
\begin{align*}
\emph{effects and \structtypes}\qquad & A, B, C &\ ::= &\ \hbox{\zcdc{Effect}}\ E \;|\; \hbox{\zcdc{AdverbDataType}}\ D \;|\; A \oplus B \\
\emph{adverb theories}\qquad & P, Q, R &\ ::= &\ \hbox{\zcdc{AdverbTheory}}\ \cong_D \;|\; P \uplus Q
\end{align*}

\textbf{Properties of $\oplus$}
\begin{align*}
\textsc{Commutativity}&\ : &\ A \oplus B \approx B \oplus A \\
\textsc{Associativity}&\ : &\ (A \oplus B) \oplus C \approx A \oplus (B \oplus C)
\end{align*}

\textbf{Properties of $\uplus$}
\begin{align*}
\textsc{Commutativity}&\ : &\ P \uplus Q \Leftrightarrow Q \uplus P \\
\textsc{Associativity}&\ : &\ (P \uplus Q) \uplus R \Leftrightarrow P \uplus (Q \uplus R) \\
\textsc{Idempotence}&\ : & P \uplus P \Leftrightarrow P
\end{align*}

\caption{The algebra for effects and composable \structs.}\label{fig:structs-algebra}
\end{figure}

\subsection{The Coq Implementation}\label{subsec:coq-impl}

All the \structtypes{} we have seen~(\cref{fig:advs}) are recursive. When we
compose these \structs, we cannot simply put these inductive types into a sum
type---we need to adapt each adverb so that it recurses on the new composed
adverb rather than itself. In other words, we need \emph{extensible inductive
types}. However, extensible inductive types are not directly supported by most
formal reasoning systems including Coq. In fact, how to support extensible
inductive types is part of an open problem known as \emph{the expression
problem}~\citep{expression}.

\begin{figure}[t]

\begin{subfigure}[b]{\textwidth}
\codeplus{ProgramAdverbs/Fix/Fix.v}{fix1}{language=Coq}
\codeplus{ProgramAdverbs/Fix/Fix.v}{fixRel}{language=Coq}
\caption{The algebra and the least fixpoint operators for effects and
\structtypes{}~(\cdc{Alg1}, \cdc{Fix1}), and for adverb theories~(\cdc{AlgRel,
FixRel}).}\label{fig:fixes}
\end{subfigure}

\begin{subfigure}[b]{\textwidth}
\begin{lstlisting}[language=Coq]
Variant Sum1 (F G : (Set -> Set) -> Set -> Set) K R :=
  Inl1 (a : F K R) | Inr1 (a : G K R).
Variant SumRel {F : Set -> Set}
        (P Q : (forall (A : Set), relation (F A)) -> forall (A : Set), relation (F A))
        (K : forall (A : Set), relation (F A)) : forall (A : Set), relation (F A) :=
| InlRel {A : Set} {a b : F A} : P K _ a b -> SumRel P Q K _ a b
| InrRel {A : Set} {a b : F A} : Q K _ a b -> SumRel P Q K _ a b.

Notation "F ⊕ G" := (Sum1 F G).
Notation "F ⊎ G" := (SumRel F G).
\end{lstlisting}
\caption{The Coq definitions for the $\oplus$ and $\uplus$ operators.}\label{fig:sums}
\end{subfigure}
\caption{Key definitions for implementing composable \structs{} in Coq.}\label{fig:fixes-sums-coq}
\end{figure}

In this paper, we address the problem and implement composable adverbs in Coq
using a technique presented in \emph{Meta Theory \`{a} la
Carte}~(MTC)~\citep{mtc}. The key idea of MTC is using Church encodings of data
types~\citep{lfix, modular-visitor} instead of Coq's native inductive types. We
apply and extend this idea to define the two least fixpoint operators \cdc{Fix1}
and \cdc{FixRel} that work on \structtypes{} and adverb theories,
respectively. We show the definitions of these operators in \cref{fig:fixes}.

We define the disjoint union $\oplus$ by first refactoring the types
of \structtypes{} and effects. We make both \structtypes{} and effects have the
type \cdc{(Set -> Set) -> Set -> Set} where the first parameter is
a \emph{recursive parameter} and the second parameter is a return type. We can
then define $\oplus$ simply as a sum type on \cdc{(Set -> Set) -> Set -> Set},
as shown in \cref{fig:sums}. Similarly, we define $\uplus$ as a sum type
on \cdc{(forall (A : Set), relation (F A)) -> forall (A : Set), relation (F A)} for
any \cdc{F : Set -> Set}.

\Cref{fig:adverb_data_types} shows the definitions of composable \structtypes{}.
Compared with the \structtypes{} in \cref{fig:advs}, a composable \structtype{}
replaces the effect parameter~(which is named \cdc{E}) with a recursive
parameter~(which is named \cdc{K}) so that it ``recurses'' on \cdc{K} instead of
itself.

We also factor out the \cdc{Pure} constructor, a common part shared by multiple
basic \structtypes, as a separate composable \structtype{}
called \cdc{ReifiedPure}. In this way, we avoid introducing multiple \cdc{Pure}
constructors, \eg by combining \cdc{Statically}
and \cdc{Conditionally}. Furthermore, we remove the \cdc{Embed} constructors in
composable \structtypes. Thanks to the uniform treatment of effects
and \structs, we can now embed effects simply by including them in \cdc{K}, so
we have no need for those constructors.

\begin{figure}[t]
\begin{subfigure}[b]{\textwidth}
\begin{lstlisting}[language=Coq]
Variant ReifiedPure (K : Set -> Set) (R : Set) : Set :=
| Pure (r : R).
Variant ReifiedFunctor (K : Set -> Set) (R : Set) : Set :=
| FMap {X : Set} (g : X -> R) (f : K X).
Variant ReifiedApp (K : Set -> Set) (R : Set) : Set :=
| LiftA2 {X Y : Set} (f : X -> Y -> R) (g : K X) (a : K Y).
Variant ReifiedSelective (K : Set -> Set) (R : Set) : Set :=
| SelectBy {X Y : Set} (f : X -> ((Y -> R) + R)) (a : K X) (b : K Y).
Variant ReifiedMonad (K : Set -> Set) (R : Set) : Set :=
| Bind {X : Set} (m : K X) (g : X -> K R).
\end{lstlisting}
\caption{The composable \structtypes.}\label{fig:adverb_data_types}
\end{subfigure}

\begin{subfigure}[b]{\textwidth}
\begin{lstlisting}[language=Coq]
Class AppKleenePlus (F : Type -> Type) `{Applicative F} :=
  { kplus {A} : F A -> F A }.
Class FunctorPlus (F : Type -> Type) `{Functor F} :=
  { plus {A} : F A -> F A -> F A }.

(* The adverb data type for Repeatedly. *)
Variant ReifiedKleenePlus (K : Set -> Set) (R : Set) : Set :=
| KPlus : K R -> ReifiedKleenePlus K R.
(* The adverb data type for Nondeterministically. *)
Variant ReifiedPlus (K : Set -> Set) (R : Set) : Set :=
| Plus : K R -> K R -> ReifiedPlus K R.
\end{lstlisting}
\caption{The \structtypes{} of \cdc{Nondeterministically} and \cdc{Repeatedly}.}\label{fig:addons}
\end{subfigure}

\caption{The composable \structtypes{} and add-on \structtypes.}
\end{figure}

As an example, we can define an ``inductive type''\ \cdc{T : Set -> Set} that is
composed of \cdc{ReifiedPure}, \cdc{ReifiedApp}, and some effect \cdc{E : (Set
-> Set) -> Set -> Set} as follows:
\begin{lstlisting}[language=Coq]
Definition T := Fix1 (ReifiedPure ⊕ ReifiedApp ⊕ E).
\end{lstlisting}
The \cdc{T} data type here is equivalent to the
non-composable \cdc{ReifiedApp} shown in \cref{fig:advs}.

Adverb interpretation can be defined as an algebra of type \cdc{Alg1 F
E}~(\cref{fig:fixes}) where \cdc{F} is the \structtype{} and \cdc{E} is the
instance we are interpreting to. To apply this ``interpretation algebra'' to the
composed ``inductive type'', we fold it over \cdc{Fix1} as follows:
\begin{lstlisting}[language=Coq]
Definition foldFix1 {E A} (alg : Alg1 F E) (f : Fix1 F A) : E A := f _ alg.
\end{lstlisting}

We define all composable \structtypes{} using \cdc{Set} rather than \cdc{Type}
because we use the impredicative sets extension in Coq, following MTC.\@ The
consequence of this decision is that (1)~certain types cannot inhabit \cdc{Set},
and (2)~the extension is inconsistent with certain set of axioms such as the
axiom of unique choice together with the law of excluded
middle.\footnote{\url{https://github.com/coq/coq/wiki/Impredicative-Set}} We
also develop other mechanisms like the injection type classes, the induction
principles following MTC.\@ We omit more detail here due to the space
constraint. The interested readers can find them in MTC or our supplementary
artifact~\citep{tlon-artifact}.

Besides MTC, there are other solutions that address the expression problem in
theorem provers like Coq. We discuss those alternative solutions
in \cref{sec:discussion}.

\subsection{Add-on Adverbs}\label{subsec:addons}

Another benefit of making \structs{} composable is that we can now define two
add-on adverbs, namely \cdc{Repeatedly} and \cdc{Nondeterministically}, which
are not suitable as standalone adverbs. These two adverbs reify two classes of
functors, namely \cdc{AppKleenePlus} and \cdc{FunctorPlus}, that we define
ourselves. We show these classes of functors and their reifications
in \cref{fig:addons}. \cdc{AppKleenePlus} is a subclass of \cdc{Applicative} and
represents the ``Kleene
plus''.\footnote{\url{https://en.wikipedia.org/wiki/Kleene_star\#Kleene_plus}}
It is a ``Kleene plus'' rather than a ``Kleene star'' because no empty element
is defined. \cdc{FunctorPlus} is similar to the commonly-used \cdc{Alternative}
and \cdc{MonadPlus} type classes in Haskell, but contains no empty element and
only requires itself to be a subclass of \cdc{Functor}. We define these type
classes' reifications as add-on adverbs so that these adverbs can be composed
with classes of functors at different expressive levels: \eg \cdc{Repeatedly}
can be composed with \cdc{Statically} as well as \cdc{Dynamically}.

We show the adverb theories of \cdc{Repeatedly} and \cdc{Nondeterministically}
in \cref{fig:addon-class}. Both of these two add-on adverbs are somewhat
nondeterministic, so one change we make to their adverb theories is adding
refinement relations~($\subseteq$) in addition to equivalence
relations~($\cong$).

\begin{figure}[t]
\centering
\begin{align*}
\textsc{Repeat} \quad  : \quad & \forall n,\ \subseteqRule{repeat a n}{kplus a} \\
\textsc{Kplus} \quad  : \quad & \inferrule{\subseteqRule{a}{kplus b}}{\subseteqRule{kplus a}{kplus b}}{} \\
\textsc{Commutativity} \quad  : \quad & \congRule{plus a b}{plus b a} \\
\textsc{Associativity} \quad  : \quad & \congRule{plus a (plus b c)}{plus (plus a b) c} \\
\textsc{Plus} \quad : \quad & \inferrule{\subseteqRule{a}{c} \\ \subseteqRule{b}{c}}{\subseteqRule{plus a b}{c}}{} \\
\textsc{Left Plus} \quad  : \quad & \subseteqRule{a}{plus a b} \\
\textsc{Right Plus} \quad  : \quad & \subseteqRule{b}{plus a b}
\end{align*}
\caption{The adverb theories for \cdc{Repeatedly}
and \cdc{Nondeterministically}. The function \cdc{repeat a n} repeats \cdc{a}
for \cdc{n} times. Functions \cdc{kplus} and \cdc{plus} are smart constructors
of \cdc{KPlus} and \cdc{Plus}, respectively.}\label{fig:addon-class}
\end{figure}

\begin{figure}[t]
\begin{lstlisting}[language=Coq]
(* FunctorPlus transformer. *)
Definition fmapPowerSet {A B : Type} (f : A -> B) (a : PowerSet I A) : PowerSet I B :=
  fun r => exists a', a a' /\ fmap f a' ≡ r.

Definition plusPowerSet {A : Type} (a b : PowerSet I A) : PowerSet I A :=
  fun r => a r \/ b r.

(* AppKleenePlus transformer. *)
Definition liftA2PowerSet {A B C : Type} (f : A -> B -> C)
           (a : PowerSet I A) (b : PowerSet I B) : PowerSet I C :=
  fun r => exists a' b', a a' /\ b b' /\ (liftA2 f a' b' ≡ r).

Fixpoint repeatPowerSet {A : Type} (a : PowerSet I A) (n : nat) : PowerSet I A :=
  match n with
  | 0 => a
  | S n => liftA2PowerSet (fun _ x => x) a (repeatPowerSet a n)
  end.

Definition kplusPowerSet {A : Type} (a : PowerSet I A) : PowerSet I A :=
  fun r => exists n, repeatPowerSet a n r.
\end{lstlisting}
\caption{The \cdc{FunctorPlus} transformer instance and the \cdc{AppKleenePlus}
transformer instance of the \cdc{PowerSet} data type. $\equiv$ is the lawful
equivalence relation on original functor/applicative
functor \cdc{I}.}\label{fig:powerset-addon}
\end{figure}

We show that these two adverbs are sound with respect to the following adverb
simulations:
\begin{align*}
\hbox{\cdc{ReifiedKleenePlus}} & \models_{\hbox{\cdc{PowerSet}}} \hbox{\cdc{AppKleenePlus}} \\
\hbox{\cdc{ReifiedPlus}} & \models_{\hbox{\cdc{PowerSet}}} \hbox{\cdc{FunctorPlus}}
\end{align*}
The definition of \cdc{PowerSet} data type is the same as that
in \cref{fig:nd_trans}, but we are using its \cdc{AppKleenePlus} transformer
and \cdc{FunctorPlus} transformer instances here. The core definitions of these
transformers are shown in \cref{fig:powerset-addon}.

\section{Examples}\label{sec:case-study}

In this section, we use two different examples to demonstrate the usefulness and
different aspects of \structs{} and \logics.

\subsection{Haxl}\label{subsec:haxl-study}

In our first example, we show that we can use composable adverbs to capture two
different computation patterns in the same library. We also demonstrate
interpreting composable adverbs to a shallow embedding in a modular way.

We illustrate these aspects via an example based on the core ideas of Haxl. Haxl
is a Haskell library developed and maintained by Meta~(formerly known as
Facebook) that automatically parallelizes certain operations to achieve better
performance~\citep{haxl}. As an example, suppose that we want to fetch data from
a database and we have a \cdc{Fetch : Type -> Type} data type that encapsulates
the fetching effect. The key insight of the Haxl library is to distinguish the
operations of \cdc{Fetch}'s \cdc{Monad} instance and those of
its \cdc{Applicative} instance. When we use \cdc{>>=} to bind two \cdc{Fetch}s,
those data fetches are sequential; but when we use \cdc{liftA2} to bind them,
those data fetches are batched and will be sent to the database together. To
achieve this, it is important that the definition of \cdc{liftA2} is not
equivalent to the ``default'' definition derived from \cdc{>>=}.

This design of Haxl poses a challenge to mixed embeddings based on freer monads
or any other basic adverbs discussed in \cref{sec:statically}, because we need
to distinguish when \cdc{Applicative} operations are used and when \cdc{Monad}
operations are used. This is exactly where composable adverbs are useful.

In this example, we assume that we already have a translation from
Haxl's \cdc{Applicative} and \cdc{Monad} operations to those operations in
Coq.\footnote{Tools like \hstocoq~\citep{hs-to-coq, containers} can be adapted
to implement the translation.} In our embedding, we use the following \cdc{T}
datatype to encode the \logic{} of a data fetching program:
\begin{lstlisting}[language=Coq]
Definition T := Fix1 (ReifiedPure ⊕ ReifiedApp ⊕ ReifiedMonad ⊕ DataEff).
\end{lstlisting}
We use \cdc{ReifiedApp} to model batched operations and the theory
of \cdc{StaticallyInParallel} to model their parallel nature. We
use \cdc{ReifiedMonad} to model sequential operations.

\begin{figure}[t]
\begin{subfigure}[b]{\textwidth}
\begin{lstlisting}[language=Coq]
Definition Update A :=  ((var -> val) -> A * nat).

Definition ret {A} (a : A) : Update A := fun map => (a, 0).
Definition bind {A B} (m : Update A) (k : A -> Update B) : Update B :=
  fun map => match m map with
             | (i, n) => match (k i map) with
                         | (r, n') => (r, n + n')
                         end
             end.
Definition liftA2 {A B C} (f : A -> B -> C) (a : Update A) (b : Update B) :
  Update C := fun map => match (a map, b map) with
                         | ((a, n1), (b, n2)) => (f a b, max n1 n2)
                         end.
Definition get (v : var) : Update val := fun map => (map v, 1).
\end{lstlisting}
\caption{The \cdc{Update} datatype.}\label{fig:state-haxl}
\end{subfigure}

\begin{subfigure}[b]{\textwidth}
\begin{lstlisting}[language=Coq]
Class AdverbAlg (D : (Set -> Set) -> Set -> Set) (I : Set -> Set) :=
  { adverbAlg : Alg1 D I }.

Instance CostApp : AdverbAlg ReifiedApp Update :=
  {| adverbAlg := fun d => match d with LiftA2 f a b => liftA2 f a b end |}.
Instance CostMonad : AdverbAlg ReifiedMonad Update :=
  {| adverbAlg := fun d => match d with Bind m k => bind m k end |}.
Instance CostPure : AdverbAlg ReifiedPure Update :=
  {| adverbAlg := fun d => match d with Pure a => ret a end |}.
Instance CostData : AdverbAlg DataEff Update :=
  {| adverbAlg := fun d => match d with GetData v => get v end |}.

Instance AdverbAlgSum D1 D2 I `{AdverbAlg D1 I} `{AdverbAlg D2 I} :
  AdverbAlg (D1 ⊕ D2) I name :=
  {| adverbAlg := fun a => match a with
                          | Inl1 a => adverbAlg a
                          | Inr1 a => adverbAlg a
                          end |}.
\end{lstlisting}
\caption{Interpretation algebras that interpret composable adverbs
and \cdc{DataEff} to \cdc{Update}. Thanks to Instance \cdc{AdverbAlgSum} and
Coq's type class inference, we can automatically get the interpretation
from \cdc{T} to \cdc{Update}.}\label{fig:interp-alg}
\end{subfigure}
\caption{The \cdc{Update} datatype and the interpretation from \cdc{T} to \cdc{Update}.}\label{fig:update}
\end{figure}

We cannot know statically how many database accesses would happen in a Haxl
program, because a program can choose to do different things depending on the
result of some data fetch. Therefore, we need to pick an effect interpretation
for \cdc{DataEff} to reason about this property. In this example, we are
assuming that the database does not change, so we interpret our \logic{} to a
shallow embedding whose semantic domain is the update monad~\citep{update}.

The key definitions of the update monad are shown in \cref{fig:state-haxl}. The
update monad is essentially a combination of a reader monad and a writer
monad. In our example, the ``reader state'' has type \cdc{var -> val} which
represents an immutable key-value database we can read from. The ``writer
state'' is a \cdc{nat}, which represents the accumulated number of database
accesses. The \cdc{bind} operation propagates the key-value database and
accumulates the cost.

Additionally, we define a \cdc{liftA2} operation, which only records the maximum
number of database accesses in one of its branches. This is not the same as
the \cdc{liftA2} operation that can be automatically derived from the monad
instance of \cdc{Update}. Furthermore, this \cdc{liftA2} is commutative. Thanks
to that, we can interpret \cdc{T} to the \cdc{Update} datatype without the
help of a \cdc{PowerSet} transformer.

\Cref{fig:interp-alg} shows how we interpret composed adverbs in a modular way.
First, we define a type class called \cdc{AdverbAlg} for interpretation
algebras. We then define an interpretation from each individual composable
adverb and effect in \cdc{T} to \cdc{Update}. Finally, the interpretation
from \cdc{T} to \cdc{Update} can be automatically inferred by Coq thanks to
the instance \cdc{AdverbAlgSum}. If we would like to add another effect or
composable adverb to \cdc{T}, we only need to add one more instance
of \cdc{AdverbAlg} and we do not need to modify any existing interpretation
algebras.

Interested readers can find the full Coq implementation of the \cdc{Update} data
type, the \cdc{AdverbAlg} type class and relevant instances, along with a few
simple examples in our supplementary artifact~\citep{tlon-artifact}.

\subsection{A Networked Server}\label{subsec:server-study}

A common technique used in formal verification is dividing the verification into
multiple layers and establishing a refinement relation between every two
layers~\citep{ccal, armada, vellvm-itree, itree-server}. This approach offers
better abstraction and modularity, as at each layer, we only need to consider
certain subsets of properties.

In this example, we show the usefulness of \structs{} and \logics{} in a layered
approach. Specifically, we define an \emph{intermediate-level} specification
that omits implementation details about execution order, \etc\@ Since the
specification is only more \emph{nondeterministic} in its control flow, we would
like our formal verification to show that an implementation refines the
specification \emph{without} interpreting effects to a shallow embedding. This
is exactly where \structs{} and \logics{} can help.

We demonstrate this vision above via a simple server adapted from that
of \citet{itree-server}. The server communicates with multiple clients via a
network interface. A client initiates a communication with the server by sending
a request that is a number. Whenever the server receives a request, it stores
the number of that request and sends back a number in its store---a client does
not necessarily receive what they sent before, because the server can interleave
multiple sessions.

We show that a specific implementation of such a server refines an
intermediate-level specification. We also show the refinements based
on \logics{} with the help of adverb theories. Unlike \citet{itree-server}, we
do not show that the implementation further refines a higher-level specification
based on observations over a network, as that is beyond the scope of this work.

\begin{figure}[t]
\hspace{1em}
\begin{subfigure}[t]{0.4\textwidth}
\begin{lstlisting}[numbers=left,frame=single]
newconn ::<- accept ;;
IF (not (*newconn == 0)) THEN
  newconn_rec ::=
    connection *newconn READING ;;
  conns ::++ newconn_rec
END ;;
FOR y IN conns DO
  IF (y->state == WRITING) THEN
    r ::<- write y->id *s ;;
    y->state ::= CLOSED
  END ;;
  IF (y->state == READING) THEN
    r ::<- read y->id ;;
    IF (*r == 0) THEN
      y->state ::= CLOSED
    ELSE
      s ::= *r ;;
      y->state ::= WRITING
    END
  END
END.
\end{lstlisting}
\caption{The implementation \texttt{Impl} in \socketnet.}\label{fig:socketnet-server}
\end{subfigure}
\hspace{0.5em}
\begin{subfigure}[t]{0.49\textwidth}
\begin{lstlisting}[frame=single]
Some
  (Or (newconn ::<- accept ;;
       IF (not (*newconn == 0)) THEN
         newconn_rec ::=
           connection *newconn READING ;;
         conns ::++ newconn_rec
       END)
       (OneOf (conns) y
         (Or (IF (y->state == WRITING) THEN
                r ::<- write y->id *s ;;
                y->state ::= CLOSED
              END)
             (IF (y->state == READING) THEN
                r ::<- read y->id ;;
                IF (*r == 0) THEN
                  y->state ::= CLOSED
                ELSE
                  s ::= *r ;;
                  y->state ::= WRITING
                END
              END))))
\end{lstlisting}
\caption{The specification \texttt{Spec} in \socketnetspec.}\label{fig:socketnet-spec}
\end{subfigure}
\caption{The implementation and the intermediate layer specification of our
networked server.}\label{fig:server}
\end{figure}

\paragraph{The implementation.}
The server is implemented using a single-process \emph{event
loop}~\citep{flash}. Instead of processing a request and sending back a response
immediately, the server divides a session with a client into multiple steps. In
each iteration of the event loop, the server advances the session of each
request by one step, thus interleaving multiple sessions.

We show the main loop body of our adapted version of the networked server
in \cref{fig:socketnet-server}. For simplicity, we use a custom language
called \socketnet. \socketnet{} supports datatypes like booleans, natural
numbers, and a special record type called \texttt{connection}. It has network
operations like \cdc{accept}, \cdc{read}, and \cdc{write}. All these operations
return natural numbers, with 0 indicating failures. The language does not have a
while loop but it has a \texttt{FOR} loop that iterates over a list. The loop
variable is implemented as a pointer that points to elements in the list
iteratively. We also use C-like notations~(\ie~\texttt{*} and \texttt{->}) for
operations on pointers.

The implementation \texttt{Impl} maintains a list of connections
called \texttt{conns}.\@ Each connection in \texttt{conns} is in one of the
three possible states: \cdc{READING}, \cdc{WRITING}, or \cdc{CLOSED}. At the
start of each loop, the server checks if there is a new connection waiting to be
established by calling the non-blocking operation \cdc{accept}. If there is, the
server creates a new \texttt{connection} with the \cdc{READING} state and adds
it to \texttt{conns}. The server then goes over each \texttt{connection}
in \texttt{conns}: if a connection is in the \cdc{READING} state, the server
tries to read from the connection and updates an internal state \texttt{s} with
the recently read value; if a connection is in the \cdc{WRITING} state, the
server sends the current value of its internal state \texttt{s} to the
connection; once a connection enters the \cdc{CLOSED} state, it remains that
state forever and the server will not do anything with it---we design the server
in this way for simplicity; a more realistic server should remove closed
connections from \texttt{conns}.

\paragraph{The specification.}

We show our specification \texttt{Spec}
in \cref{fig:socketnet-spec}. \texttt{Spec} is written in a language
called \socketnetspec. \socketnetspec{} adds a few additional commands
to \socketnet: \texttt{Some} is an unary operation that models the ``Kleene
plus''; \texttt{Or} is a binary operation that models a nondeterministic choice
wrapped inside a ``Kleene plus''; \texttt{OneOf} is like \cdc{Or}, but it
nondeterministically chooses from a list---line 8 means that we
nondeterministically assign the variable \texttt{y} with one element from the
list in \texttt{conns}.

\texttt{Spec} is more nondeterministic compared with \texttt{Impl}. At each
iteration of the event loop, \texttt{Impl} always first tries to \texttt{accept}
a connection. It then goes over the list of \texttt{conns} in a fixed
order. \texttt{Spec} does not enforce order: an \texttt{accept} could happen
immediately after another \texttt{accept}; we can access elements
in \texttt{conns} in any order and some connection might get visited more often
than others.

\paragraph{\logics{} and the refinement proof.}

To show that \texttt{Impl} refines \texttt{Spec}, we embed both \socketnet{}
and \socketnetspec{} in Coq using \structs. We use the following datatype:
\begin{lstlisting}[language=Coq]
Definition T := Fix1 (ReifiedKleenePlus ⊕ ReifiedPlus ⊕ ReifiedPure ⊕ ReifiedMonad ⊕
  NetworkEff ⊕ MemoryEff ⊕ FailEff).
\end{lstlisting}
We have already seen the first four adverbs. Effect \cdc{NetworkEff} models the
effects incurred by network operations \texttt{accept}, \texttt{read},
and \texttt{write}. Effect \cdc{MemoryEff} models the effects incurred by
assigning values to variables and retrieving values from them. Finally,
effect \cdc{FailEff} models when the program crashes.

\ifextended
We use $\den{\cdot}_{\text{\cdc{T}}}^L$ to denote a language $L$'s \logic{}
in \cdc{T}. For the sake of space, we only show how we
embed \texttt{Some}, \texttt{Or}, and \texttt{OneOf} here:
\begin{align*}
\den{\texttt{Some c}}_{\text{\zcdc{T}}}^S &= \hbox{\zcdc{kplus}}\ \den{\hbox{\zcdc{c}}}_{\text{\zcdc{T}}}^S
\\ \den{\texttt{Or c1 c2}}_{\text{\zcdc{T}}}^S &= \hbox{\zcdc{kplus (plus}}\ \den{\hbox{\zcdc{c1}}}_{\text{\zcdc{T}}}^S\ \den{\hbox{\zcdc{c2}}}_{\text{\zcdc{T}}}^S\hbox{\zcdc{)}}
\\ \den{\texttt{OneOf xs y c}}_{\text{\zcdc{T}}}^S &= \hbox{\zcdc{get c >>= (fun xs => kplus (foldr}}
\\ & \qquad \hbox{\zcdc{(fun v s => plus (set y v >>}}\ \den{\hbox{\zcdc{c}}}_{\text{\zcdc{T}}}^S\ \hbox{\zcdc{) s)}}
\\ & \qquad \hbox{\zcdc{(pure tt) xs))}}
\end{align*}
\cdc{Some} is simply a \cdc{kplus}~(from \cdc{Repeatedly}). \cdc{Or} is
a \cdc{plus}~(from \cdc{Nondeterministically}) wrapped inside
a \cdc{kplus}. \cdc{OneOf xs y c} is a bit complicated: we first
use \cdc{get}, an effectful \cdc{MemoryEff} operation that retrieves a
value from the reference \cdc{xs} in the memory, to get a list, which
we also call \cdc{xs} and it shadows the other \cdc{xs}; we then fold
the list nondeterministically many times using \cdc{plus}
over \cdc{xs}; each operand joined by a \cdc{plus} is a \cdc{set y v},
an effectful \cdc{MemoryEff} operation that set the value \cdc{v} in
the reference \cdc{y} in the memory, followed by the embedding of
command \cdc{c}.
\else
We use $\den{\cdot}_{\text{\cdc{T}}}^I$ to denote \socketnet's \logic{}
in \cdc{T} and $\den{\cdot}_{\text{\cdc{T}}}^S$ to
denote \socketnetspec's \logic{} in \cdc{T}. For the sake of space, we omit
the embeddings here.
\fi

\begin{figure}[t]
\begin{subfigure}[t]{0.3\textwidth}
\begin{lstlisting}[language=Coq,mathescape=true]
Definition L1 :=
  $A$ ;;
  FOR y IN conns DO
    $B$ ;;
    $C$ ;;
  END.
\end{lstlisting}
\end{subfigure}
\begin{subfigure}[t]{0.3\textwidth}
\begin{lstlisting}[language=Coq,mathescape=true]
Definition L2 :=
  $A$ ;;
  OneOf (conns) y
    ($B$ ;; $C$).
\end{lstlisting}
\end{subfigure}
\begin{subfigure}[t]{0.3\textwidth}
\begin{lstlisting}[language=Coq,mathescape=true]
Definition L3 :=
  $A$ ;;
  OneOf (conns) y
    (Or $B$ $C$).
\end{lstlisting}
\end{subfigure}
\caption{Program \cdc{L1} written in \socketnet, and programs \cdc{L2} and \cdc{L3} written in \socketnetspec.}\label{fig:layers}
\end{figure}

We would like to show that
$\den{\texttt{Impl}}_{\text{\cdc{T}}}^I \subseteq \den{\texttt{Spec}}_{\text{\cdc{T}}}^S$.
Recall that $\subseteq$ is the refinement relation
on \structs~(\cref{subsec:addons}). The theorem states that the \logic{} of our
implementation \texttt{Impl} in \cdc{T} refines the \logic{} of our
specification \texttt{Spec} in \cdc{T}.

To show that, we first observe that \texttt{Impl} and \texttt{Spec} share some
common program fragments, \eg~lines 1--6 of \texttt{Impl} are the same as lines
2--7 of \texttt{Spec}. Indeed, there are three such common fragments and we name
them $A$~(lines 1--6 of \texttt{Impl}), $B$~(lines 8--11 of \texttt{Impl}), and
$C$~(lines 12--20 of \texttt{Impl}), respectively. We then define three
programs \cdc{L1}, \cdc{L2}, and \cdc{L3} shown in \cref{fig:layers}. These
programs represent some intermediate layers between \texttt{Impl}
and \texttt{Spec}. We prove the following theorem:
\begin{theorem}\label{thm:refinement}
$\den{\texttt{Impl}}_{\text{\cdc{T}}}^I \subseteq \den{\texttt{L1}}_{\text{\cdc{T}}}^I \subseteq \den{\texttt{L2}}_{\text{\cdc{T}}}^S \subseteq \den{\texttt{L3}}_{\text{\cdc{T}}}^S \subseteq \den{\texttt{Spec}}_{\text{\cdc{T}}}^S$.
\end{theorem}
\begin{proof}
We show
$\den{\texttt{Impl}}_{\text{\cdc{T}}}^I \subseteq \den{\texttt{L1}}_{\text{\cdc{T}}}^I$
by associativity of \cdc{Dynamically}. Both
$\den{\texttt{L1}}_{\text{\cdc{T}}}^I \subseteq \den{\texttt{L2}}_{\text{\cdc{T}}}^S$
and
$\den{\texttt{L2}}_{\text{\cdc{T}}}^S \subseteq \den{\texttt{L3}}_{\text{\cdc{T}}}^S$
can be proven by an induction over \texttt{conns} and with the help of theories
of \cdc{Dynamically}, \cdc{Repeatedly} and \cdc{Nondeterministically}. Finally,
we prove
$\den{\texttt{L3}}_{\text{\cdc{T}}}^S \subseteq \den{\texttt{Spec}}_{\text{\cdc{T}}}^S$
by the theories of \cdc{Dynamically}, \cdc{Repeatedly},
and \cdc{Nondeterministically}.
\end{proof}

Interested readers can find the full Coq implementation
of \socketnet, \socketnetspec, the \logics{} of these two languages, the
implementation \texttt{Impl}, the specification \texttt{Spec}, as well as the
full proof of \cref{thm:refinement} in our supplementary
artifact~\citep{tlon-artifact}.

\section{Discussion}\label{sec:discussion}\label{sec:compare-freer}

\paragraph{The expression problem}

The composable \structs{} require extensible inductive types. We implement this
feature in Coq by using the Church encodings of datatypes, following the
precedent work of MTC~\citep{mtc}. There are several consequences of using
Church encodings instead of Coq's original inductive datatypes.

First, we cannot make use of Coq's language mechanisms, libraries, and plugins
that make use of Coq's inductive types (\eg~Coq's builtin induction principle
generator, the Equations plugin~\citep{coq-equations}, the QuickChick
plugin~\citep{gen-inductive, quickchick-ind},~\etc). Furthermore, the extra
implementation overheads incurred by Church encodings~(\eg~proving an algebra is
a functor, proving the induction principle using dependent types,~\etc) can be
huge. However, this situation can be helped by developing tools or plugins for
supporting Church encodings.

The other consequence is that, following the practice of MTC, we use Coq's
impredicative set extension. This causes two problems: (1)~Certain types cannot
inhabit \cdc{Set}, and (2)~our Coq development is inconsistent with certain set
of axioms such as the axiom of unique choice together with the law of excluded
middle, as we have discussed in \cref{subsec:coq-impl}.

There are alternative methods for addressing the expression problem. One option
is the meta-programming approach proposed by \citet{coq-a-la-carte}. In this
approach, we can define each composable adverb separately in a meta language and
use a language plugin to generate a combined definition in Coq. This approach
does not fully address the expression problem as extending the combined
definition requires recompilation---but the amount of code that needs to be
recompiled is much smaller and the generated code uses Coq's builtin inductive
types. Another option that has recently been explored
by \citet{family-poly-for-proofs} is adding \emph{family
polymorphism}~\citep{family-poly} to theorem provers. These works are promising.
Unfortunately, they either lack mature tool support or is still in development
at the moment. We would like to explore these approaches in the future and
composable program adverbs might provide a good application to these approaches.

\paragraph{Reified vs. free structures}

Even though the reified structures used in \structtypes{} are free structures,
they are different from those free structures present in~\citet{freer, free-ap,
selective, selective-impl}. The biggest difference between reified structures
and these free structures are the parameters they recurse on: all the reified
structures recurse on both their computational parameters, while each free
structure only recurses on one of them.\footnote{With the exception of
reified/free functors, since each of them has only one computational parameters
to be recursed on.} For example, comparing \cdc{FreerMonad} in \cref{fig:freer}
and \cdc{ReifiedMonad} in \cref{fig:advs}: \cdc{FreerMonad} only recurses on the
parameter \cdc{k} of \cdc{Bind}, while \cdc{ReifiedMonad} recurses on both
parameters \cdc{m} and \cdc{k}. This means that a free structure does not just
reify a class of functors, it also converts the reification to a left- or
right-associative normal form.

One advantage of the normal forms in free structure definitions is that the type
class laws can be automatically derived from definitional equality (with the
help of the axiom of functional extensionality). However, this conversion would
eliminate some differences in the syntax. Taking \cdc{ReifiedApp} as an example,
normalizing it would result in a ``list'' rather than a ``binary tree'', making
analyzing the depth of the tree impossible. Preserving the original tree
structure of \cdc{StaticallyInParallel} also plays a crucial role in our
examples shown in \cref{sec:applicative-embedding} and~\ref{subsec:haxl-study}.

\ifextended
Another note is that the commonly used definition of free monads in Haskell
cannot be encoded in Coq because it is not strictly
positive~\citep{one-monad}. The common ways to work around this problem are:
(1)~using containers~\citep{one-monad}, or (2)~using their \emph{freer}
variants~\citep{freer, itree, mcbride-free, steelcore}.
\fi

\section{Related work}\label{sec:related-work}

\paragraph{Semantic embeddings}

There are various works that study different semantic
embeddings. \Citet{embedding} are the pioneers who coined terms such as semantic
embeddings, shallow embeddings, and deep embeddings. It is known that there are
many styles of embeddings between shallow and deep embeddings, but there is not
an agreed term on describing them. In this paper, we use the term \emph{mixed
embeddings}, which is borrowed from \citet{adam-binder}, where it is used to
describe an embedding based on freer monads. Another term \emph{deeper shallow
embeddings} is proposed by \citet{deeper-shallow}, which shows a way of
deepening any shallow embedding.

\paragraph{Freer Monads and Variants}

Freer monads~\citep{freer} and their variants are studied by many researchers in
formal verification to reason about programs with effects. Earlier work includes
the study of the \emph{delay monad}~\citep{delay} and \emph{resumption
monads}~\citep{resumption}. More recent work includes \citet{freespec2}, where
the authors use free monads to develop a modular verification framework based on
effects and effect handlers called FreeSpec.\@
\Citet{verify-effectful-haskell-in-coq} develop a framework based on free monads
and containers~\citep{containers-cat} for reasoning about Haskell programs with
effects. \Citet{pt-semantics} interpret free monads into a predicate transformer
semantics that is similar to Dijkstra monads; \Citet{reaching-star} interprets
free monads using separation logic.

On the \emph{coinductive} side, \Citet{itree} develop a coinductive variant of
freer monads called \emph{interaction trees} that can be used to reason about
general recursions and nonterminating programs. \Citet{itree-server} encode
interaction trees in VST~\citep{vst} to reason about networked
servers. \Citet{vst-certikos-connection} use interaction trees as a lingua
franca to interface and compose higher-order separation logic in VST and a
first-order verified operating system called
CertiKOS~\citep{certikos}. \Citet{gpaco} propose a technique called generalized
parameterized coinduction for developing equational theory for reasoning about
interaction trees. \Citet{vellvm-itree} use interaction trees to define a
modular, compositional, and executable semantics for
LLVM.\@ \Citet{itree-layerd} further extend the modularity of interaction trees
by extending them with layered monadic interpreters. \Citet{dijkstra-forever}
connect interaction trees with Dijkstra monads~\citep{maillard2019} for writing
termination sensitive specifications based on uninterpreted effects. \Citet{c4}
use interaction trees to verify transactional objects.  \Citet{itree-hol} apply
interaction trees to Isabelle/HOL to produce a verification and simulation
framework for state-rich process languages, which is used
by \citeauthor{robochart-itree} to give an operational semantics to RoboChart, a
timed and probabilistic domain-specific language for
robotics~\citep{robochart-itree}.

Among many variants of freer monads, one particular structure closely
resembles \structs. That is the action trees defined in \citet{steelcore}. The
action trees have four constructors, \cdc{Act}, \cdc{Ret}, \cdc{Par},
and \cdc{Bind}, whose types correspond to
effects, \cdc{ReifiedPure}, \cdc{ReifiedApp}, and \cdc{ReifiedMonad} in
composable \structs, respectively, another evidence that \structs{} are general
models. In contrast to our work, compositionality and extensibility of
``adverbs'' are not the main issue action trees try to address, so action trees
are not built in a composable way. On the other hand, action trees are embedded
with separation logic assertions, which are not the focus of \structs{}
or \logics.

\paragraph{Other Free Structures}

Other free structures are also explored by various works. \Citet{free-ap}
propose two variants of freer applicative functors, which correspond to the
left- and right-associative variants, respectively. \Citet{free-ap-coq} explores
defining freer applicative functors in Coq, and points out that the right
associative variant is harder to define in Coq. \Citet{free-monoidal} discusses
deriving free monoidal functors.\ifextended\footnote{Monoidal functors are
equivalent to applicative functors, so they also correspond to
the \cdc{Statically} adverb.}\fi \Citet{selective-impl} defines the free
selective functors.

\ifextended
In the context of formal reasoning, one of the main inspirations of our work
is \citet{free-ap}. They observe that the structures of freer monads are not
amenable to static reasoning and propose freer applicative functors. Our work
takes the observation further and identifies a class of \structs~(and
composable adverbs).
\fi

\paragraph{Programming Abstractions}

We are not the first to observe that monads are too dynamic for certain
applications. For example, \citet{non-monadic-parser} identify that a parser
that has some static features cannot be defined as a monad. Inspired by their
observation, \citet{arrows} proposes a new abstract interface called arrows. The
relationship among arrows, applicative functors, monads are studied
by \citet{idioms-arrows-monads}. \Citet{selective-parser} observe that monads
generate dynamic structures that are hard to optimize. They further show that,
by using applicative and selective functors instead, it is possible to implement
staged parser combinators that generate efficient parsers.
\ifextended
\Citet{build} observe that the datatype of tasks in a build
system~(called \cdh{Task} in their paper) can be parameterized by a class
constraint to describe various kinds of build tasks.  For example, a \cdh{Task
Applicative} describes tasks whose dependencies are determined \emph{statically}
without running the task; and a \cdc{Task Monad} describes tasks
with \emph{dynamic} dependencies.
\else
\Citet{build} observe that the datatype of tasks in a build
system can be parameterized by a class constraint to describe various kinds of
build tasks.
\fi

\section{Conclusion}

In this paper, we compare different styles of semantic embeddings and how they
impact formal reasoning about programs with effects. We find that, if used
properly, mixed embeddings can combine benefits of both shallow and deep
embeddings, and be effective in (1)~preserving syntactic structures of original
programs, (2)~showing general properties that can be proved without assumptions
over external environment, and (3)~reasoning about properties in specialized
semantic domains.

We propose \emph{\structs} and \emph{\logics}, a class of structures and a style
of mixed embeddings based on these structures, that enable us to reap these
benefits. Like free monads, \structs{} embed pure computations shallowly and
effects deeply~(and abstractly, but can later be interpreted). However,
various \structs{} correspond to alternative computation patterns, and can be
composed to model programs with multiple characteristics.

Based on \structs, \logics{} cover a wide range of programs and allow us to
reason about syntactic properties, semantic properties, and general semantic
properties with no assumption over external environment within the same
embedding.

\section*{Acknowledgment}
We thank the anonymous ICFP'22 reviewers, as well as the POPL'22 and ESOP'22
reviewers, whose valuable feedback helped improve this paper. We thank various
researchers for their feedback during their discussions with the authors,
including Stephanie Balzer, Conal Elliott, Yannick Forster, Paul He, Apoorv
Ingle, Ende Jin, Konstantinos Kallas, Anastasiya Kravchuk-Kirilyuk, Andrey
Mokhov, Benjamin C. Pierce, Nick Rioux, Andre Scedrov, Antal Spector-Zabusky,
Caleb Stanford, Kathrin Stark, Nikhil Swamy, Val Tannen, Steve Zdancewic, Weixin
Zhang, and Yizhou Zhang,~\etc\@ We thank the anonymous ICFP'22 artifact
reviewers for their comments and suggestions that helped improve the
artifact. We thank Paolo Giarrusso (on CoqClub Zulipchat), Li-yao Xia, and Irene
Yoon for helping with some Coq issues the authors encountered when developing
the artifact.

This material is based upon work supported by
the \grantsponsor{GS100000001}{National Science
Foundation}{http://dx.doi.org/10.13039/100000001} under Grant
No.~\grantnum{GS100000001}{1521539}, Grant No.~\grantnum{GS100000001}{1703835},
and Grant No.~\grantnum{GS100000001}{2006535}.  Any opinions, findings, and
conclusions or recommendations expressed in this material are those of the
authors and do not necessarily reflect the views of the National Science
Foundation.

\bibliography{ref}

\end{document}